\documentclass[english,11pt]{article}

\usepackage{wrapfig}
\usepackage{libertine}
\usepackage{geometry}
\usepackage{enumitem}
\PassOptionsToPackage{vlined, boxruled}{algorithm2e}
\usepackage{algorithm2e}
\usepackage{mleftright}
\usepackage{babel}

\usepackage{amsmath}
\usepackage{amsthm}
\usepackage{amssymb}
\usepackage{graphicx}
\usepackage{tabularx}
\usepackage{array}
\usepackage{soul}
\usepackage[dvipsnames]{xcolor}
\usepackage{mdframed}
\usepackage{xparse}
\usepackage[capitalise, noabbrev]{cleveref}

\usepackage{tcolorbox}
\usepackage{tikz,lipsum}
\tcbuselibrary{skins,breakable}

\makeatletter
\def\special{1}
\def\specialproof{1}
\def\specialdefinition{1}
\def\specialremark{1}

\def\highlight{0}
\def\sidebar{1}

\theoremstyle{plain}
\newtheorem{thm1}{Theorem}[section]
\theoremstyle{remark}
\newtheorem{pthm1}[thm1]{Theorem}
\theoremstyle{plain}
\newtheorem{lem1}[thm1]{Lemma}
\theoremstyle{plain}
\newtheorem{obs1}[thm1]{Observation}
\theoremstyle{plain}
\newtheorem{inv1}[thm1]{Invariant}
\theoremstyle{plain}
\newtheorem{cor1}[thm1]{Corollary}
\theoremstyle{definition}
\newtheorem{defn1}[thm1]{Definition}
\theoremstyle{plain}
\newtheorem{fact1}[thm1]{Fact}
\theoremstyle{remark}
\newtheorem{rem1}[thm1]{Remark}
\theoremstyle{plain}
\newtheorem{prop1}[thm1]{Proposition}
\theoremstyle{plain}
\newtheorem{asmp1}[thm1]{Assumption}

\crefname{thm1}{Theorem}{Theorems}
\crefname{pthm1}{Theorem}{Theorems}
\crefname{lem1}{Lemma}{Lemmas}
\crefname{obs1}{Observation}{Observations}
\crefname{inv1}{Invariant}{Invariants}
\crefname{cor1}{Corollary}{Corollaries}
\crefname{defn1}{Definition}{Definitions}
\crefname{fact1}{Fact}{Facts}
\crefname{rem1}{Remark}{Remarks}
\crefname{prop1}{Proposition}{Propositions}
\crefname{asmp1}{Assumption}{Assumptions}
\crefname{subsection}{Subsection}{Subsections}

\ifx\proof\undefined
\newenvironment{proof}[1][\protect\proofname]{\par
\normalfont\topsep6\p@\@plus6\p@\relax
\trivlist
\itemindent\parindent
\item[\hskip\labelsep\scshape #1]\ignorespaces
}{%
\endtrivlist\@endpefalse
}
\providecommand{\proofname}{Proof}
\fi

\newenvironment{thm}[1][]{%
\begin{thm1}[#1]%
}{\end{thm1}%
}
\newenvironment{lem}[1][]{%
\begin{lem1}[#1]%
}{\end{lem1}%
}
\newenvironment{obs}[1][]{%
\begin{obs1}[#1]%
}{\end{obs1}%
}
\newenvironment{inv}[1][]{%
\begin{inv1}[#1]%
}{\end{inv1}%
}
\newenvironment{fact}[1][]{%
\begin{fact1}[#1]%
}{\end{fact1}%
}
\newenvironment{rem}[1][]{%
\begin{rem1}[#1]%
}{\end{rem1}%
}
\newenvironment{pthm}[1][]{%
\begin{pthm1}[#1]%
}{\end{pthm1}%
}
\newenvironment{cor}[1][]{%
\begin{cor1}[#1]%
}{\end{cor1}%
}
\newenvironment{defn}[1][]{%
\begin{defn1}[#1]%
}{\end{defn1}%
}
\newenvironment{asmp}[1][]{%
\begin{asmp1}[#1]%
}{\end{asmp1}%
}
\newenvironment{prop}[1][]{%
\begin{prop1}[#1]%
}{\end{prop1}
}%

\if 1
    \if 1\highlight
        \renewenvironment{thm}[1][]{%
        \begin{mdframed}[nobreak=false,backgroundcolor=Aquamarine!60]\begin{thm1}[#1]%
        }{\end{thm1}\end{mdframed}%
        }
        \renewenvironment{lem}[1][]{%
        \begin{mdframed}[nobreak=false,backgroundcolor=YellowGreen!60]\begin{lem1}[#1]%
        }{\end{lem1}\end{mdframed}%
        }
        \renewenvironment{obs}[1][]{%
        \begin{mdframed}[nobreak=false,backgroundcolor=Salmon!60]\begin{obs1}[#1]%
        }{\end{obs1}\end{mdframed}%
        }

        \renewenvironment{prop}[1][]{%
        \begin{mdframed}[backgroundcolor=Goldenrod!60]\begin{prop1}[#1]%
        }{\end{prop1}\end{mdframed}%
        }
        
    \fi
    \if 1\sidebar
        \tcolorboxenvironment{thm}{blanker,breakable,left=4mm, before skip=10pt,after skip=10pt, borderline west={2mm}{0pt}{Aquamarine}}
        \tcolorboxenvironment{lem}{blanker,breakable,left=4mm, before skip=10pt,after skip=10pt, borderline west={2mm}{0pt}{YellowGreen}}
        \tcolorboxenvironment{obs}{blanker,breakable,left=4mm, before skip=10pt,after skip=10pt, borderline west={2mm}{0pt}{Salmon}}
        \tcolorboxenvironment{inv}{blanker,breakable,left=4mm, before skip=10pt,after skip=10pt, borderline west={2mm}{0pt}{Salmon}}
        \tcolorboxenvironment{fact}{blanker,breakable,left=4mm, before skip=10pt,after skip=10pt, borderline west={2mm}{0pt}{Salmon}}
        \tcolorboxenvironment{pthm}{blanker,breakable,left=4mm, before skip=10pt,after skip=10pt, borderline west={2mm}{0pt}{Salmon}}
        \tcolorboxenvironment{prop}{blanker,breakable,left=4mm, before skip=10pt,after skip=10pt, borderline west={2mm}{0pt}{Goldenrod}}
        \tcolorboxenvironment{asmp}{blanker,breakable,left=4mm, before skip=10pt,after skip=10pt, borderline west={2mm}{0pt}{Goldenrod}}
    \fi
\fi

\if 1\specialproof
    \if 1\highlight
        \expandafter\let\expandafter\oldproof\csname\string\proof\endcsname
        \let\oldendproof\endproof
        \renewenvironment{proof}[1][\proofname]{%
        \begin{mdframed}[nobreak=false,backgroundcolor=lightgray!60]\oldproof[#1]%
        }{\oldendproof\end{mdframed}}
    \else
        \if 1\sidebar
        \tcolorboxenvironment{proof}{
        blanker,breakable,left=4mm,
        before skip=10pt,after skip=10pt,
        borderline west={2mm}{0pt}{lightgray}}
        \fi
    \fi
\fi

\if 1\specialdefinition
    \if 1\highlight
        \renewenvironment{defn}[1][]{%
        \begin{mdframed}[innerbottommargin=0.1cm,innertopmargin=0.1cm,backgroundcolor=Apricot!60]\begin{defn1}[#1]%
        }{\end{defn1}\end{mdframed}%
        }
    \else
        \if 1\sidebar
        \tcolorboxenvironment{defn}{
        blanker,breakable,left=4mm,
        before skip=10pt,after skip=10pt,
        borderline west={2mm}{0pt}{Apricot}}
        \fi
    \fi
\fi

\if 1\specialremark
    \if 1\highlight
        
    \else
        \if 1\sidebar
        \tcolorboxenvironment{rem}{
        blanker,breakable,left=4mm,
        before skip=10pt,after skip=10pt,
        borderline west={2mm}{0pt}{Salmon}}
        \fi
    \fi
\fi

\@ifundefined{showcaptionsetup}{}{%
\PassOptionsToPackage{caption=false}{subfig}}
\usepackage{subfig}

\usepackage{float}
\floatstyle{boxed}
\restylefloat{figure}

\SetKwInOut{Input}{Input}
\SetKwProg{Initialization}{Initialization}{}{}
\SetKwProg{Fn}{Function}{}{end}
\SetKwProg{EFn}{Event Function}{}{end}
\SetKwFunction{UponJobRelease}{UponJobRelease}
\SetKwFunction{Process}{Process}
\SetKwFunction{UponHeavyF}{UponHeavyF}
\SetKwFunction{Rotate}{Rotate}
\SetKw{And}{and}
\SetKw{Or}{or}
\SetKw{Break}{break}
\SetKw{Continue}{continue}

\SetCommentSty{mycommfont}

\SetFuncSty{myfuncsty}

\LinesNumbered \RestyleAlgo{boxruled}

\date{}

\makeatother

\makeatletter
\newcommand{\algorithmfootnote}[2][\footnotesize]{%
\let\old@algocf@finish\@algocf@finish
\def\@algocf@finish{\old@algocf@finish
\leavevmode\rlap{\begin{minipage}{\linewidth}
#1#2
\end{minipage}}%
}%
}
\makeatother

\definecolor{darkred}{RGB}{200,0,0}

\newcommand{\opt}{\operatorname{OPT}}

\newcommand{\alg}{\operatorname{ALG}}

\newcommand{\prob}{{\sc sppt}\xspace}

\newcommand{\pr}[1]{\mleft(#1\mright)}

\newcommand{\pb}[1]{\mleft[#1\mright]}

\newcommand{\pc}[1]{\mleft\{#1\mright\}}

\newcommand{\ps}[1]{\mleft|#1\mright|}

\newcommand{\ceil}[1]{\mleft\lceil#1\mright\rceil}

\newcommand{\floor}[1]{\mleft\lfloor#1\mright\rfloor}

\newcommand{\IR}[2]{\mleft[#1,#2\mright)}

\newcommand{\cset}[2]{\pc{#1\middle|#2}}

\newcommand{\base}[2]{X(#1,#2)}

\newcommand{\wbase}[2]{\beta(#1,#2)}

\newcommand{\cls}[1]{\ell_{#1}}

\newcommand{\pt}[1]{p_{#1}}

\newcommand{\pts}[1]{p(#1)}

\newcommand{\wt}[1]{w_{#1}}

\newcommand{\wts}[1]{w(#1)}

\NewDocumentCommand{\Wtgen}{s!D[]{}!D<>{}!d()}{
    \IfBooleanTF{#1}{
        \IfNoValueTF{#4}{
            W^{*#3}_{#2}
        }{
            W^{*#3}_{#2}\pr{#4}
        }
    }{
        \IfNoValueTF{#4}{
            W^{#3}_{#2}
        }{
            W^{#3}_{#2}\pr{#4}
        }
    }
}

\NewDocumentCommand{\Vgen}{s!D[]{}!D<>{}!d()}{
    \IfBooleanTF{#1}{
        \IfNoValueTF{#4}{
            V^{*#3}_{#2}
        }{
            V^{*#3}_{#2}\pr{#4}
        }
    }{
        \IfNoValueTF{#4}{
            V^{#3}_{#2}
        }{
            V^{#3}_{#2}\pr{#4}
        }
    }
}

\NewDocumentCommand{\reqsgen}{s!D[]{}!D<>{}!d()}{
    \IfBooleanTF{#1}{
        \IfNoValueTF{#4}{
            Q^{*#3}_{#2}
        }{
            Q^{*#3}_{#2}\pr{#4}
        }
    }{
        \IfNoValueTF{#4}{
            Q^{#3}_{#2}
        }{
            Q^{#3}_{#2}\pr{#4}
        }
    }
}

\newcommand{\reqs}[1]{\reqsgen(#1)}

\NewDocumentCommand{\Wt}{sm}{
    \IfBooleanTF{#1}{
        \Wtgen<*>(#2)
    }{
        \Wtgen(#2)
    }
}

\newcommand{\rpt}[2]{y_{#1}(#2)}

\newcommand{\rpts}[2]{y(#1,#2)}

\newcommand{\rptso}[2]{y^*(#1,#2)}

\newcommand{\ept}[1]{\tilde{p}_{#1}}

\newcommand{\SC}[1][]{{\mathrm{SC}}_{#1}}

\newcommand{\ope}[1][]{1+\epsilon_{#1}}

\newcommand{\xbar}[1]{#1\in \mathbb{R}^+\cup \pc{0}}

\NewDocumentCommand{\xpt}{om}{\IfNoValueTF{#1}{x_{#2}}{x_{#2}(#1)}}

\NewDocumentCommand{\score}{om}{\IfNoValueTF{#1}{\Lambda_{#2}}{\Lambda_{#2}({#1})}}

\newcommand{\rpto}[2]{y^*_{#1}(#2)}

\newcommand{\cont}[4][]{\gamma^{#1}_{#2}(#3,#4)}

\NewDocumentCommand{\cov}{omm}{\IfNoValueTF{#1}{B\pr{#2,#3}}{B_{#1}\pr{#2,#3}}}

\newcommand{\cbc}{\theta}

\newcommand{\cbp}{\theta'}

\newcommand{\p}[2][]{\pi_{#2}^{#1}}

\newcommand{\diff}[1]{\,\mathrm{d}#1}

\newcommand{\dstr}{\mu}

\newcommand{\bround}{\lambda}

\begin{document}
\title{Flow Time Scheduling with Uncertain Processing Time}
\author{%
\begin{tabular}{>{\centering\arraybackslash}p{4.5cm}>{\centering\arraybackslash}p{4.5cm}>{\centering\arraybackslash}p{4.5cm}}
Yossi Azar\thanks{Supported in part by the Israel Science Foundation (grant No. 2304/20 and grant No. 1506/16).} & Stefano Leonardi\thanks{Supported by the ERC Advanced Grant 788893 AMDROMA ``Algorithmic and Mechanism Design Research in Online Markets'' and MIUR PRIN project ALGADIMAR ``Algorithms, Games, and Digital Markets''.} & Noam Touitou\tabularnewline
\textsf{\small{}azar@tau.ac.il} & \textsf{\small{}leonardi@diag.uniroma1.it} & \textsf{\small noamtouitou@mail.tau.ac.il}\tabularnewline
{\small{}Tel Aviv University} & {\small{}Sapienza University of Rome} & {\small{}Tel Aviv University}\tabularnewline
\end{tabular}
}
\maketitle

\begin{abstract}
We consider the problem of online scheduling on a single machine in order to minimize weighted flow time.
The existing algorithms for this problem (STOC '01, SODA '03, FOCS '18) all require exact knowledge of the processing time of each job.
This assumption is crucial, as even a slight perturbation of the processing time would lead to polynomial competitive ratio.
However, this assumption very rarely holds in real-life scenarios.

In this paper, we present the first algorithm for weighted flow time which do \emph{not} require exact knowledge of the processing times of jobs.
Specifically, we introduce the Scheduling with Predicted Processing Time (\prob) problem, where the algorithm is given a prediction for the processing time of each job, instead of its real processing time.
For the case of a constant factor distortion between the predictions and the real processing time, our algorithms match \emph{all} the best known competitiveness bounds for weighted flow time -- namely $O(\log P), O(\log D)$ and $O(\log W)$, where $P,D,W$ are the maximum ratios of processing times, densities, and weights, respectively.
For larger errors, the competitiveness of our algorithms degrades gracefully.
\end{abstract}
\thispagestyle{empty} \clearpage \setcounter{page}{1}

\section{Introduction}
\label{sec:Intro}

The field of online scheduling focuses on efficient processing of jobs by machines, where the jobs are not known in advance but are released over time.
An algorithm in this setting must assign each job to a machine, which must then process the job for some specific amount of time (called the \emph{processing time} of the job).

In a classic setting for online scheduling, the algorithm is given a single machine, and aims to minimize the \emph{total flow time}: the sum over jobs of the time that the job was pending in the algorithm (i.e. the difference between its completion time and its release time).
In this setting, we usually allow \emph{preemption} of a currently-processed job, which is pausing its processing until some future time, in which the processing will resume from the state in which it was paused.
A classic result by Smith from 1956~\cite{NAV:NAV3800030106} shows that this problem can be solved optimally using the SRPT (shortest remaining processing time) schedule.

A natural generalization of the total flow time metric is total \emph{weighted} flow time, in which each job also has a \emph{weight}.
The goal is therefore to minimize the weighted sum of flow times, where the weight of a job is the weight of its flow time in the goal function.
While natural, this problem has proved much harder than minimizing unweighted flow time: the best known algorithms~\cite{DBLP:journals/talg/BansalD07,DBLP:conf/focs/AzarT18} have competitive ratios with logarithmic dependence on various parameters of the input; such parameters include the ratio $P$ of the largest to smallest processing time in the input, the ratio $W$ of the largest to smallest job weight in the input, and the ratio $D$ of the largest to smallest density in the input (where the density of a job is the ratio of its weight to its processing time).
As shown by Bansal and Chan~\cite{DBLP:conf/soda/BansalC09}, a dependence on these parameters is necessary (without added leniency such as speed augmentation).

However, these algorithms for both weighted and unweighted flow time make the assumption that the processing time of a job becomes known to the algorithm upon the release of the job.
This assumption almost never holds in real-world scenarios, as nearly all computer programs of some complexity have varying running times.
In the case of weighted flow-time, this assumption is ubiquitous -- it is crucial to all known algorithms of sub-polynomial competitiveness~\cite{DBLP:conf/stoc/ChekuriKZ01,DBLP:journals/talg/BansalD07,DBLP:conf/focs/AzarT18}, as replacing the real processing time with the predicted processing time in those algorithms would yield polynomial competitive ratios (e.g. $\Omega(P)$) even for a slight misprediction by a factor of $1+\epsilon$ (compare this to the original polylogarithmic competitiveness of those algorithms).

\paragraph{Predicted Processing Time and Robustness.} While knowing the exact processing time is infeasible in a real-world scenario, one could hold a \emph{prediction} for the processing time of a job, of varying accuracy.
Such predictions might be obtained from machine-learning algorithms, simple heuristics, or any other source.
A good algorithm in this setting would be able to perform well given only the prediction, rather than the actual procesing time, upon the release of a job.

The fact that the known guarantees for scheduling problems apply only for algorithms that know the processing time completely creates a gap between theory and practice.
In this paper, we address this gap.
Namely, we introduce the model of \emph{scheduling with predicted processing time} (\prob), in which upon the release of a job, we are provided with a prediction of the job's processing time.
An instance of the problem is characterized by a \emph{distortion parameter}, which is the maximum ratio over jobs of the job's real processing time and its predicted processing time.

In this paper, we present competitive algorithms which are \emph{$\dstr$-robust}; that is, they maintain their competitive ratio guarantees for all inputs with distortion parameter at most $\dstr$.
The competitive ratios have a polynomial dependency on $\dstr$ (contrast  with previously-known algorithms, where the competitive ratio immediately degrades to e.g. $\Omega(P)$  even when $\dstr = 1+\epsilon$ for a small $\epsilon$).
We assume that a value for the single parameter $\dstr$ that bounds the distortion of the input can be learned, in order to apply the correct robust algorithm; the case in which this assumption does not hold seems surprisingly tricky, and is discussed in \cref{sec:Disc}.


At this point, one might ask whether the definition of the distortion parameter as the \emph{maximum} distortion is indeed the correct measure for the error of the prediction; for example, could guarantees in some average of errors, e.g. geometric, be given?
However, our choice of maximum distortion is not overly-conservative, but rather prescribed by the nature of the flow-time metric -- specifically, due to its sensitivity to bad local competitive ratio.
Consider for example, a scheduling instance which initially releases jobs with very erroneous, which would cause the algorithm to perform badly at some point in time.
The adversary could then maintain this bad perform by releasing a stream of very short jobs; this well-known ``bombardment'' technique would guarantee a bad (global) competitive ratio for the algorithm.
However, the processing time of these short ``bombardment'' jobs could be predicted with full accuracy, yielding a geometric average of errors which is arbitrarily close to $1$ as the ``bombardment'' continues; this would also be the case for every other reasonable average.



\subsection{Our Results}
\paragraph{The {\prob} Problem.} We present algorithms for {\prob} for minimizing weighted flow time.
An algorithm, coupled with some competitiveness guarantee, is called \emph{$\dstr$-robust} if its competitiveness guarantee is achieved for any instance with distortion parameter at most $\dstr$ (and not only for inputs where the prediction matches the real processing time).


Without loss of generality, we assume one-sided error: that is, that the real processing time of each job is at least the predicted processing time, and at most $\dstr$ times the predicted processing time.
This one-sided error can be trivially obtained from a two-sided error with distortion parameter $\dstr'$ by dividing each prediction by the distortion parameter $\dstr'$; the new distortion parameter $\dstr$ of the new one-sided-error instance would be at most $(\dstr')^2$.

Recall that $P,D,W$ are the ratios of maximum to minimum processing time, density and weight, respectively.
For any $\dstr$, our results are:

\begin{enumerate}
    \item A $\dstr$-robust algorithm for {\prob} with weighted flow time, which is simultaneously $O(\dstr^3 \log(\dstr P))$-competitive and $O(\dstr^3 \log (\dstr  D))$-competitive.

    \item A $\dstr$-robust algorithm for {\prob} with weighted flow time, which is $O(\dstr^2 \log W)$-competitive.
\end{enumerate}

In the process of designing the $O(\dstr^2 \log W)$-competitive algorithm, we also design a $\dstr$-robust algorithm for {\prob} with \emph{unweighted} flow time, which is $2\ceil {{\dstr}^2}$-competitive.
Note that, somewhat surprisingly, even for $\dstr$ arbitrarily close to $1$, there exists a lower bound of $2$ on the competitiveness of any deterministic $\dstr$-robust algorithm, which we show in \cref{sec:UWLB}.

For the weighted setting, and for a constant $\dstr$, our algorithms match the best known results in terms of all three parameters: the $O(\log W)$-competitive algorithm of Bansal and Dhamdhere~\cite{DBLP:journals/talg/BansalD07}, and the $O(\log P)$ and $O(\log D)$ competitive algorithms in \cite{DBLP:conf/focs/AzarT18}.
Even for distortion significantly larger than constant, e.g. $\dstr = \Theta(\log P)$, our algorithms still obtain polylogarithmic competitiveness.

A special case of {\prob} is the problem of semiclairvoyant scheduling~\cite{DBLP:conf/soda/BenderMR02,DBLP:journals/tcs/BecchettiLMP04}.
In this problem, instead of getting the processing time $\pt{q}$ of a  job $q$, we are given its \emph{class}, which is $\ell_q = \floor{\log_{\rho}\pt{q}}$ for some constant $\rho > 1$.
Applying our algorithms for \prob we obtain the first semiclairvoyant algorithms for weighted flow time, which match the best known guarantees for the \emph{clairvoyant} setting.
As a side result, we also obtain an improvement to the best known result for unweighted flow time.
The exact statement of our results for the semiclairvoyant setting are given in \cref{sec:SCS}.

\paragraph{Paper Structure.} The $\dstr$-robust algorithm with logarithmic dependency on either $P$ or $D$ is presented in \cref{sec:CLP}.
The $\dstr$-robust algorithm with logarithmic dependency on $W$ is shown in \cref{sec:LogW}. \Cref{sec:LogW} comprises two subsections: \cref{subsec:UW} presents a $\dstr$-robust algorithm for the unweighted setting, while \cref{subsec:LogW} uses the algorithm of \cref{subsec:UW} to construct the algorithm for the weighted setting.
\Cref{sec:UWLB} shows a lower bound of $2$-competitiveness for $\dstr$-robust algorithms in the unweighted setting, even for vanishingly small distortion.
\Cref{sec:SCS} discusses the application of the algorithms for {\prob} to the semiclairvoyant setting.

\subsection{Our Techniques}
In developing robust algorithms for \prob, we have used orthogonal methods and analyses for the algorithm for $P$ and $D$ and the algorithm for $W$.
Curiously, it seems that the techniques of each single algorithm could not be applied to obtain the results of the other algorithm.

\textbf{The $O(\dstr^3\log (\dstr P))$-competitive algorithm for \prob}, which is also $O(\dstr^3 \log (\dstr D) )$-competitive, is a simply-stated and novel algorithm.
The flavor of its analysis is somewhat reminiscent of the $O(\log^2 P)$-competitive algorithm of Chekuri \emph{et al.}~\cite{DBLP:conf/stoc/ChekuriKZ01} (though, of course, our analysis shows an improved competitiveness, as our guarantee is $O(\log P)$ when $\dstr$ is constant).

The algorithm begins by rounding the weights of jobs up to powers of $\Theta(\dstr)$; this is the only place in the algorithm where $\dstr$ is used.
In weighted flow time, an algorithm must choose whether to prioritize high-weight jobs or cost-effective jobs (high density).
The algorithm contends with this issue in the following way: it chooses the maximum weight $w$ of a living job, and then chooses a job from the highest density class such that this class contains at least $w$ weight.
Inside this density class, the algorithm always chooses the highest weight job.

In the analysis of this algorithm, we show that the maximum weight in the algorithm never exceeds the \emph{total} weight in the optimal solution.
As for the lower weight classes which the algorithm \emph{does} use, a volume-based analysis shows that the total weight in those classes is bounded.

\textbf{The $O(\dstr^2   \log W)$-competitive algorithm for {\prob} with weighted flow time} is completely different in both methods and analysis from the previous algorithm, instead using the analysis framework introduced by Bansal and Dhamdhere \cite{DBLP:journals/talg/BansalD07}.

The algorithm is constructed in two steps.
First, we introduce a $\dstr$-robust algorithm for minimizing \emph{unweighted} flow time.
This algorithm uses two bins to maintain pending jobs, a \emph{full} bin and a \emph{partial} bin, which contain roughly the same number of jobs.
Only jobs in the partial bin are processed by the algorithm, so that the jobs in the full bin retain their original processing time.
Thus, the full bin provides a ``counterweight'' of full jobs to the partial bin; the intuition for this is that the algorithm wants to limit the fraction of partially-processed jobs at any point in time.
Indeed, naive processing without attempting to limit the number of partially-processed jobs would result in unbounded competitive ratio, as observed in~\cite{DBLP:journals/tcs/BecchettiLMP04}.
The competitiveness proof of this algorithm bears similarities to the proof framework of Bansal and Dhamdhere~\cite{DBLP:journals/talg/BansalD07}, also found in~\cite{DBLP:conf/focs/AzarT18}.
However, both of those papers use SRPT as a crucial component in their algorithms and analyses, which is infeasible in {\prob}.
Thus, our algorithms are designed to bypass this requirement.

When designing the priority of jobs in the full bin, we need to overcome the fact the algorithm does not know the remaining processing times of jobs (prohibiting strategies such as SRPT).
Instead, the algorithm identifies and eliminates \emph{violations}, which are pairs of jobs \emph{provably} not ordered according to SRPT; this occurs when their predicted processing times differ by more than $\dstr$.
The algorithm solves violations using a \emph{rotation} operation, which identifies all violations involving a job and rotates their priorities.
The resulting priority sequence, which contains no violations, is \emph{``quasi-SRPT''} -- that is, the upper bound for the processing time of a low-priority job is always higher than the lower bound for the processing time of a high-priority job.

Interestingly, the priority among the jobs in the \emph{partial} bin is LIFO, which is perhaps counterintuitive.
While we are not sure that LIFO is the only option for achieving competitiveness, our proof crucially requires this feature.

We use the robust algorithm for unweighted flow time as a component in the $O(\dstr^2 \log W)$-competitive algorithm for {\prob} with weighted flow time.
This algorithm is composed of multiple full-partial bin pairs, one for each of the $\log W$ weight classes.
Each bin pair behaves as an instance of the unweighted algorithm, and the algorithm chooses the heaviest bin pair to process at any given time (more specifically, the bin pair with the heaviest partial bin is processed).
The proof framework used for the unweighted algorithm is now utilized, as it naturally supports multiple bin pairs.


One could hope that \emph{any} $\dstr$-robust algorithm for the unweighted setting could be extended (at a loss of $O(\log W)$) to the weighted setting through binning.
However, this is not known to be the case for weighted flow time -- the binning technique is \textbf{not} black-box, and previous algorithms using binning (\cite{DBLP:journals/talg/BansalD07,DBLP:conf/focs/AzarT18}) demanded specific properties from the algorithm used in each bin.
The main focus in designing the unweighted component of our algorithm is for it to ``play nice'' with the binning technique; this fact requires a delicate assignment of job priorities.

\subsection{Related Work}
The first algorithm with polylogarithmic guarantee for minimizing weighted flow time on a single machine was presented by Chekuri \emph{et al.}~\cite{DBLP:conf/stoc/ChekuriKZ01}, which was $O(\log^2 P)$-competitive.
Bansal and Dhamdhere~\cite{DBLP:journals/talg/BansalD07} gave an $O(\log W)$-competitive algorithm.
Bansal and Chan~\cite{DBLP:conf/soda/BansalC09} then showed that any deterministic algorithm for the problem must be \scalebox{0.9}{$\Omega\pr{\min\pc{\sqrt{\frac{\log W}{\log\log W}},\sqrt{\frac{\log \log P}{\log \log \log P}}}}$}-competitive, showing that dependence on at least one parameter is necessary.
In~\cite{DBLP:conf/focs/AzarT18}, $ O(\log P) $-competitive and $O(\log D)$-competitive algorithms were given, where $D$ is the maximum ratio of densities of jobs (where density is processing time over weight).
For more than a single machine, the problem is essentially intractable -- Chekuri \emph{et al.}~\cite{DBLP:conf/stoc/ChekuriKZ01} show a lower bound of \scalebox{0.9}{$\Omega\pr{\min\pc{\sqrt{P},\sqrt{W},(n/m)^{\frac{1}{4}}}}$}.
With $(1+\epsilon)$ speed augmentation, this problem is significantly easier~\cite{DBLP:journals/ipl/KimC03a,DBLP:journals/talg/BansalD07,DBLP:journals/mst/ZhuCL15}.

The field of combining online algorithms with machine-learned predictions has seen significant interest in the last few years, with some works involving predictions in scheduling.
For example, Purohit \emph{et al.}~\cite{purohit2018improving} studied minimizing job completion time with predicted processing time.
Other forms of prediction are also used; for example, Lattanzi \emph{et al.}~\cite{LattanziLMV20} studied the case of restricted assignment to multiple machines to minimize makespan, and used an algorithm-specific prediction (machine weights).
Scheduling with predictions has also been studied in the queueing theory setting~\cite{mitzenmacher:LIPIcs:2020:11699,DBLP:journals/corr/abs-2006-15463}.
Additional work on algorithms with predictions can be found in \cite{DBLP:conf/icml/LykourisV18, DBLP:conf/nips/MedinaV17, mitzenmacher2020algorithms}

The case in which no prediction is given (and thus the algorithm knows nothing about the processing time of jobs) is called the \emph{nonclairvoyant} model, and was studied in~\cite{Motwani1994,Becchetti2001,DBLP:journals/algorithmica/BansalDKS04,Im2014,Im2017}.
For minimizing unweighted flow time on a single machine, an $O(\log n)$-competitive randomized algorithm is given in~\cite{Becchetti2001}, where $n$ is the number of jobs.
This matches the lower bound of $\Omega(\log n)$ for randomized algorithms in~\cite{Motwani1994}.
Note that randomization is needed, and deterministic algorithms cannot get a sub-polynomial guarantee~\cite{Motwani1994}.
We are not aware of any nonclairvoyant results for minimizing weighted flow time.

Bechetti \emph{et al.}~\cite{DBLP:journals/tcs/BecchettiLMP04} presented a $13$-competitive algorithm for the \emph{semiclairvoyant} model, with the goal of minimizing unweighted flow time.
The semiclairvoyant setting was also considered by Bender \emph{et al.}~\cite{DBLP:conf/soda/BenderMR02} for minimizing stretch (where the stretch of a job is the ratio of its flow time to its length)

\section{Preliminaries}
\label{sec:Prelim}

We first formalize the scheduling problem we consider in this paper.
Then, we formalize the {\prob} model, and the prediction given to the algorithm.

\paragraph*{The Scheduling Problem.}
In the scheduling problem we consider, jobs arrive over time.
Each job has its own processing time, which is the time it must be processed by the machine in order to be completed.
The algorithm may choose at any point in time which job to process, and is allowed preemption.

An input consists of a set of jobs $Q$.
Each job $q\in Q$ has the following properties:
\begin{enumerate}[itemsep=0.1em,topsep=0.1em]
    \item A \emph{processing time} $\pt{q}>0$.

    \item A \emph{weight} $\wt{q}>0$.

    \item A \emph{release time} $r_q$.
\end{enumerate}

The goal of the algorithm $\alg$ is to minimize the weighted flow time  $F^{\alg} := \sum_{q\in Q} \wt{q}\cdot \pr{c_q^{\alg} - r_q}$
where $c_q^{\alg}$ is the time in which request $q$ is completed in the algorithm.
An equivalent definition, which is the dominant one in this paper, is $F^{\alg} = \int_0^{\infty} \pr{\sum_{q\in Q^{\alg}(t)}\wt{q}} \diff{t}$
where $Q^{\alg}(t)$ is the set of pending jobs in the algorithm at time $t$.

\paragraph*{The {\prob} Model.}

In the {\prob} model, the previous scheduling problem has the following modification.
When a job $q$ arrives, the algorithm does not become aware of $\pt{q}$.
Instead, it is given $\ept{q}$, the predicted processing time of the job.

An algorithm for \prob, coupled with a competitiveness guarantee, is called $\dstr$-robust if it maintains its competitiveness guarantee for inputs in which
$\ept{q} \le \pt{q} <\dstr \ept{q}$  for every job $q$ in the input.


\paragraph*{Notation.}
We denote by $\rpt{q}{t}$ the \emph{remaining processing time} of $q$ at $t$ -- that is, $\pt{q}$ minus the amount of time already spent processing $q$ until time $t$.
This amount is also called the \emph{volume} of $q$ at $t$.

For a set of jobs $Q'$, we define:
\begin{itemize}[itemsep=0.1em,topsep=0.1em]
    \item $\wts{Q'} := \sum_{q\in Q'} \wt{q}$.
    \item $\pts{Q'} := \sum_{q\in Q'} \pt{q}$.
    \item $\rpts{Q'}{t} := \sum_{q\in Q'} \rpt{q}{t} $.
\end{itemize}


For any time $t$:
\begin{itemize}[itemsep=0.1em,topsep=0.1em]
    \item We define $\reqsgen(t)$ to be the set of pending jobs at $t$ in the algorithm.
    \item We define $\delta(t):=|\reqsgen(t)|$ to be the number of pending jobs at $t$ in the algorithm.
    \item We define $\Vgen(t):= \rpts{\reqsgen(t)}{t}$ to be the total volume of pending jobs at $t$ in the algorithm.
    \item We define $\Wtgen(t):= \wts{\reqsgen(t)}$ to be the total weight of pending jobs at $t$ in the algorithm.
\end{itemize}

%

Fixing the optimal solution $\opt$ for the given input, we refer to the attributes of $\opt$ using the superscript~* (e.g. $\rpto{q}{t},\reqsgen*(t),\delta^{*}(t),\Vgen*(t),\Wtgen*(t)$).

\section{Weighted Flow Time -- Ratio of Processing Times}
\label{sec:CLP}

\newcommand{\lastt}[1]{t_{#1}}

\NewDocumentCommand{\reqp}{smmm}{
    \IfBooleanTF{#1}{
        \reqsgen*[#2,#3](#4)
    }{
        \reqsgen[#2,#3](#4)
    }
}

\newcommand{\nclass}{\Lambda}

\NewDocumentCommand{\Wtp}{smmmO{}}{
    \IfBooleanTF{#1}{
        \Wtgen[#2,#3]<*#5>(#4)
    }{
        \Wtgen[#2,#3]<#5>(#4)
    }
}



\NewDocumentCommand{\WtpP}{mmmo}{\Wtp{#1}{#2}{#3}[\textrm{p}]}

\NewDocumentCommand{\WtpF}{mmmo}{\Wtp{#1}{#2}{#3}[\textrm{f}]}




\NewDocumentCommand{\Vp}{smmmO{}}{
    \IfBooleanTF{#1}{
        \Vgen[#2,#3]<*#5>(#4)
    }{
        \Vgen[#2,#3]<#5>(#4)
    }
}



\newcommand{\edn}[1]{\tilde{\text{id}}_{#1}}

\newcommand{\edc}[1]{\operatorname{idc}\pr{#1}}

\newcommand{\wc}[1]{\operatorname{wc}\pr{#1}}

\newcommand{\qw}[1]{u_{#1}}

\newcommand{\cwc}[1]{\operatorname{cwc}\pr{#1}}

\newcommand{\clw}[1]{\operatorname{clw}\pr{#1}}

\newcommand{\id}{{\sc ei}-density\xspace}

\NewDocumentCommand{\owc}{o}{\IfNoValueTF{#1}{i^*}{i^* \pr{#1}}}

In this section, we present a $\dstr$-robust algorithm which is $O(\dstr^3\log (\dstr P))$-competitive.

The main result of this section is the following theorem.
\begin{thm}
    \label{thm:CLP_Competitiveness}
    For every $\dstr$, there exists a $\dstr$-robust algorithm for the \prob problem that is $O(\dstr^3\log (\dstr P))$-competitive.
\end{thm}

\subsection{The Algorithm}

We start by rounding the weights of jobs up to powers of $\bround = \Theta(\dstr)$; specifically, we choose $\bround = 16\dstr + 6$.
From now on, when referring to the weights of jobs, we refer to those rounded weights.

\begin{defn}[weights and ID]
    Throughout this section we use the following definitions:
    \begin{itemize}
        \item We define $\wc{q}:=\log_{\bround}(\wt{q})$ to be the \emph{weight class} of a request $q$ (which is an integer, due to the rounding of weights).
        \item We define $\qw{i} := \bround^{i}$ for every integer $i$.
        \item We define $\edc{q} := \floor{\log_2 \pr{\frac{\ept{q}}{\wt{q}}}}$ to be the \emph{estimated inverse density class} of a job $q$, abbreviated as \id class (the amount $\frac{\ept{q}}{\wt{q}}$ is called the \id of $q$).
    \end{itemize}
\end{defn}

\begin{defn}[predicated jobs]
    \label{defn:CLP_Qualified}
    For every two predicates $p_1,p_2$ we define
    \[
        \reqp{p_1}{p_2}{t} = \cset{q\in \reqs{t}}{p_1(\wc{q})\text{ and }p_2(\edc{q})}
    \]
    Similarly, we define $\Wtp{p_1}{p_2}{t} = \wts{\reqp{p_1}{p_2}{t}}$ and $\Vp{p_1}{p_2}{t} = \rpts{\reqp{p_1}{p_2}{t}}{t}$.
\end{defn}

For example, $\reqp{\le i}{=j}{t}$ is the set of all living jobs $q$ at $t$ such that $\wc{q} \le i$ and $\edc{q} = j$.
We use $\top$ to denote the predicate that always evaluates to true; it is used when we would like to take jobs of every \id/weight.


We say that a job is \emph{partial} in the algorithm if it has been processed for any amount of time; otherwise, the job is \emph{full}.

\paragraph{Description of the Algorithm.} The algorithm we describe prefers high weight jobs and high density jobs, and attempts to balance these two considerations.
This is done by processing some job of the maximum possible density, i.e. minimum \id, subject to the total weight of jobs from that \id class exceeding the weight of the heaviest living job.
After choosing the correct \id class in this manner, the algorithm must choose a job from that \id class; it chooses from the heaviest weight class in which there is a living job.
Once the correct \id and weight classes have been chosen, the algorithm prefers partial jobs, thus maintaining the fact that there is at most one partial job of each \id-weight combination.

The formal description of the algorithm appears in \cref{alg:CLP_Algorithm}.

\begin{figure*}[h]
    \begin{center}
        \includegraphics[width=0.7\textwidth]{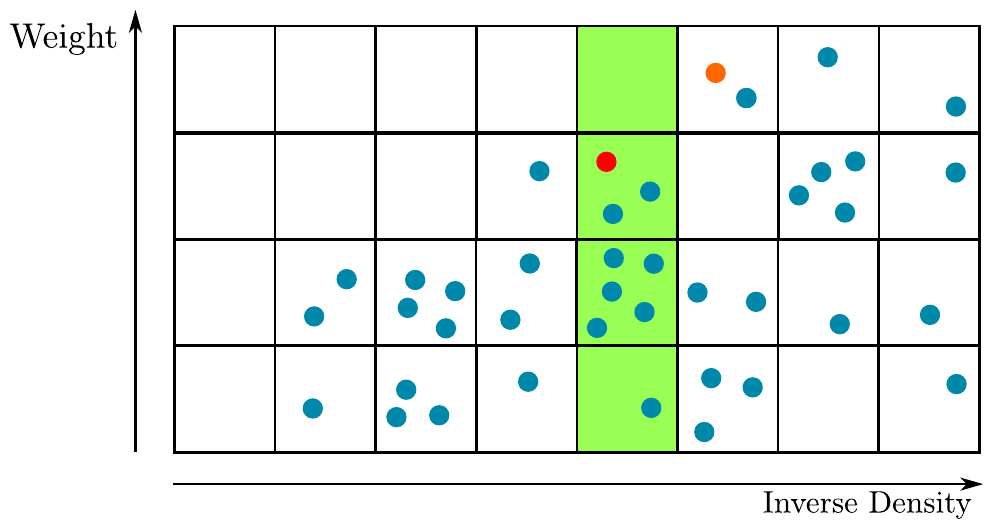}
%
    \end{center}
    This figure shows a possible state of the algorithm, visualized as a table of weight classes and inverse density classes.
    Each job appears as a point inside the appropriate entry in the table.
    The algorithm chooses the maximum weight class of a living job -- one such maximum weight job is shown in orange.
    The algorithm then chooses the minimum \id class (i.e. maximum density) in which the total weight of jobs is at least the weight of that orange job -- that chosen \id class is colored green.
    Finally, the algorithm chooses a maximum weight job inside the green \id class for processing, and this job is shown in red.
    \caption[]{\label{fig:CLP_State}The State and Operation of \cref{alg:CLP_Algorithm}}
\end{figure*}

\begin{algorithm}
    \renewcommand{\algorithmcfname}{Algorithm}
    \caption{\label{alg:CLP_Algorithm} Scheduling with Predictions -- Weighted}

    \BlankLine

    \EFn(\tcp*[h]{at any point in time $t$}){\Process{}}{
        Let $i$ be the largest weight class in which a job is alive.
%
%

        Let $j$ be the minimal \id class such that $\Wtp{\top}{=j}{t} \ge \qw{i}$.

        Let $i'$ be the maximum weight class in which a job is alive in $\reqp{\top}{=j}{t}$.

        Process a job from $\reqp{=i'}{=j}{t}$ (preferring a partial job if exists).
    }

\end{algorithm}

\subsection{Analysis}
\label{subsec:CLP_Analysis}

Consider an instance in which the distortion parameter is at most $\dstr$.
Observe that the ratio between the maximum predicted processing time and the minimum predicted processing time is at most $\dstr P$.

Fixing any weight $\qw{i}$, observe that a job of weight $\qw{i}$ can only belong to \id classes in a limited range, which contains at most  $\ceil{\log(\dstr P)} + 1$ classes.
We henceforth define $\nclass := \ceil{\log(\dstr P)} + 1$, the maximum number of \id classes to which jobs of a specific weight can belong.



The main lemma used to prove \cref{thm:CLP_Competitiveness} is the following, which states that the algorithm is locally competitive.

\begin{lem}
    \label{lem:CLP_LocalCompetitiveness}
    At any point in time $t$, it holds that $\Wt{t} \le O(\nclass \dstr^2) \cdot \Wt*{t}$.
\end{lem}

The main focus henceforth would be on proving  \cref{lem:CLP_LocalCompetitiveness}.

\begin{defn}[analogue of \cref{defn:CLP_Qualified} for $\opt$]
    For every two predicates $p_1,p_2$ we define
    \[
        \reqp*{p_1}{p_2}{t} = \cset{q\in \reqsgen*(t)}{p_1(\wc{q})\text{ and }p_2(\edc{q})}
    \]
    We define $\Wtp*{p_1}{p_2}{t} = \wts{\reqp*{p_1}{p_2}{t}}$ and $\Vp*{p_1}{p_2}{t} = \rptso{\reqp*{p_1}{p_2}{t}}{t}$.

    In addition, we use $\Delta$ to denote to volume difference between the algorithm and the optimal solution; for example, $\Delta \Vp{\le i}{=j}{t} = \Vp{\le i}{=j}{t} - \Vp*{\le i}{=j}{t} $.
\end{defn}

\begin{defn}[important weight classes]
    We define the following:
    \begin{enumerate}
        \item We define $\cwc{t}$ to be the weight class of the job processed by the algorithm at time $t$.
        \item We define $\clw{t}$ to be the largest weight class of a living job in the algorithm at $t$.
        \item We define $\owc[t]$ to be the minimum weight class such that $\qw{\owc[t]} > \frac{4\dstr}{\bround} \cdot \Wt*{t}$.
    \end{enumerate}
\end{defn}

The following observation is immediate from the above definition.
\begin{obs}
    There is no job of weight class more than $\owc[t]$ alive in $\opt$ at time $t$.
\end{obs}

Henceforth, fix a point in time $t$ for proving \cref{lem:CLP_LocalCompetitiveness}.
For brevity, we also write $\owc$ instead of $\owc[t]$.

The following lemma allows converting weight to volume, and its proof appears in \cref{sec:CLP_Proofs}.

\begin{lem}
    \label{lem:CLP_VolumeConversion}
    Let $i$ be a weight class, and let $j_1,j_2$ be two \id classes such that $j_1 \le j_2$.
    Then
    \[
        \sum_{j=j_1}^{j_2} \Wtp{\le i}{=j}{t}
        \le \min\pc{2(j_2-j_1+1), 4\nclass}\cdot \qw{i} + 2\dstr \Wt*{t} + 2\sum_{j=j_1}^{j_2} \max\pc{0,\frac{\Delta \Vp{\le i}{\le j}{t}}{2^j}}
    \]
\end{lem}

\begin{defn}
    Define $\lastt{j}$ to be the last time prior to $t$ in which the algorithm processed a job of weight class at most $\owc$ and \id class strictly more than $j$.
\end{defn}

\Cref{prop:CLP_DeltaAndWeightBounds} shows that the algorithm has low maximum weight at $\lastt{j}$ for every $j$, and bounds differences in volume between the algorithm and the optimal solution (which enables use of \cref{lem:CLP_VolumeConversion}).
The proof of \cref{prop:CLP_DeltaAndWeightBounds} is given in \cref{sec:CLP_Proofs}.

\begin{prop}
    \label{prop:CLP_DeltaAndWeightBounds}
    Let $j$ be an \id class, then it holds that:
    \begin{enumerate}
        \item $\clw{\lastt{j}} \le \owc$
        \item $\Delta \Vp{\le \owc}{\le j}{t} \le \dstr\cdot 2^{j+2} \qw{\clw{\lastt{j}}} $
    \end{enumerate}
\end{prop}

%
%

\begin{proof}[Proof of \cref{lem:CLP_LocalCompetitiveness}]
    First, we claim that $\clw{t} \le \owc$.
    Assume otherwise, and observe the last time $t'$ in which the algorithm processed a job of weight at most $\owc$.
    It must be that $\clw{t'} > \owc$; this is since the optimal solution has no jobs of weight class more than $\owc$ alive at $t$, and since the algorithm works only on such jobs from $t'$ onwards.
    However, $t'$ is also equal to $\lastt{j}$ for some $j$, which contradicts \cref{prop:CLP_DeltaAndWeightBounds}.
    Thus, it must be that $\clw{t} \le \owc$.

    Let $j_{\min}$ be the minimum \id class in which a job of weight class at most $\owc$ can exist, and let $j_{\max}$ be the maximum possible \id class.
    It holds that
    \begin{align*}
        \Wt{t}
        &= \sum_{j=j_{\min}}^{j_{\max}} \Wtp{\le \owc}{=j}{t} \\
        &\le 4\nclass \qw{\owc} + 2\dstr\Wt*{t} + \sum_{j=j_{\min}}^{j_{\max}} 2\max\pc{0, \frac{\Delta \Vp{\le \owc}{\le j}{t}}{2^j}}\\
        &\le 4\nclass \qw{\owc} + 2\dstr\Wt*{t} + \sum_{j=j_{\min}}^{j_{\max}} 8\dstr\cdot \qw{\clw{\lastt{j}}}
    \end{align*}

    where the equality is due to $\clw{t} \le \owc$, the first inequality is due to \cref{lem:CLP_VolumeConversion} and the second inequality follows from \cref{prop:CLP_DeltaAndWeightBounds}.
    Now, consider that for every \id class $j$, at $\lastt{j}$ we processed a job $q$ of \id class $h$ where $h > j$, such that $\wc{q} \le \clw{\lastt{j}}$.
    There also existed at $\lastt{j}$ a job $q'$ such that $\wc{q'} = \clw{\lastt{j}}$; denote its \id class $h':=\edc{q'}$.
    It holds that $h' \ge h$, otherwise $q'$ would have been chosen for processing at $\lastt{j}$ instead of $q$.
    Thus, it must be the case that an \id class larger than $j$ is one of the $\nclass$ \id classes to which $\clw{\lastt{j}}$ can belong.

    From the preceding argument, $\owc = \clw{\lastt{j}}$ can hold for at most $\nclass$ values of $j$; $\owc-1 = \clw{\lastt{j}}$ can hold for at most $\nclass$ additional values of $j$; and so on.
    It thus holds that:
    \[
        8\dstr\sum_{j=j_{\min}}^{j_{\max}} \qw{\clw{\lastt{j}}} \le 8\dstr \nclass \sum_{k=0}^{\infty} \qw{\owc - k} \le 16\dstr \nclass \qw{\owc}
    \]

    Therefore, we have
    \begin{equation}
        \label{eq:CLP_LocalCompetitivenessBound}
        \Wt{t} \le \nclass (16\dstr + 4) \qw{\owc} + 2\dstr\Wt*{t}
    \end{equation}

    Observe that $\qw{\owc} \le 4\dstr\Wt*{t}$ -- this stems from the definition of $\owc$.
    Plugging into \cref{eq:CLP_LocalCompetitivenessBound},
    we get that 
    \[
        \Wt{t} \le O(\nclass \dstr^2) \Wt*{t} \qedhere
    \]
\end{proof}

\begin{proof}[Proof of \cref{thm:CLP_Competitiveness}]
   Plugging the definition of $\nclass$ into \cref{lem:CLP_LocalCompetitiveness}, we have that
    \[
        \Wt{t} \le O(\dstr^2 \log (\dstr P) ) \Wt*{t}
    \]
    Recall that this is after the rounding of the weights to powers of $\bround = \Theta(\dstr)$; taking this rounding into account, the algorithm is $O(\dstr^3 \log (\dstr P))$-competitive.
\end{proof}

\subsection{Ratio of Densities}
\label{subsec:CLP_Density}

Recall that the density of a job $q$ is $\frac{\wt{q}}{\pt{q}}$.
We show that the dependence on $P$ in the competitiveness of \cref{alg:CLP_Algorithm} (as stated in \cref{thm:CLP_Competitiveness}) can be replaced with a dependence on the parameter $D$, which is the maximum ratio of the densities of two jobs in the input.
Concretely, we show that the following theorem holds.

\begin{thm}
    \label{thm:CLP_LogDCompetitiveness}
    \Cref{alg:CLP_Algorithm} is a $\dstr$-robust, $O(\dstr^3 \log (\dstr D))$-competitive algorithm for the \prob problem.
\end{thm}

The proof of \cref{thm:CLP_LogDCompetitiveness} appears in \cref{sec:LogD_Proofs}, and is almost immediate from the previous proof of \cref{thm:CLP_Competitiveness}.

%
%

\section{Weighted Flow Time -- Ratio of Weights}
\label{sec:LogW}
In this section, we present the $O(\dstr^2 \log W)$-competitive $\dstr$-robust algorithm for {\prob} with weighted flow time.
We do so in two steps: first, we introduce $\dstr$-robust algorithm for the unweighted setting, which is $2\ceil {{\dstr}^2}$-competitive.
Then, we show that instances of this unweighted algorithm can be combined and applied to the weighted setting, at an additional loss of an $O(\log W)$ factor in competitiveness.

\subsection{The Algorithm for Unweighted Flow Time in the Prediction Model}
\label{subsec:UW}

In this subsection, we present an algorithm in the prediction model for the unweighted setting -- that is, where $\wt{q} = 1$ for all $q\in Q$.

\subsubsection{The Algorithm}
The algorithm holds each pending job in one of two bins, a ``full'' bin $F$ and a ``partial'' bin $P$.
The algorithm attempts to divide the total weight of the pending jobs equally between the two bins.
The algorithm refers to the pending jobs in $F$ by $Q_F$ and the pending jobs in $P$ by $Q_P$.
For ease of notation, we also define $\delta_F = \ps{Q_F}$ and $\delta_P = \ps{Q_P}$.


The course a job undergoes in the algorithm consists of being initially assigned to bin $F$, then moving to bin $P$, and finally being processed in bin $P$ until its eventual completion.
The algorithm maintains a bijective priority mapping $\p{F}:Q_F \to \pc{1,2,\cdots,\delta_F}$ from the living jobs $Q_F$ to the natural integers from $1$ to $\delta_F$.
This mapping determines the priority of the jobs of $F$.
That is, whenever moving a job from $F$ to $P$ is required, the algorithm moves the highest priority job ($q$ such that $\p{F}(q)=\delta_F$).

The algorithm also maintains (implicitly) a priority mapping $\p{P}:Q_P \to \pc{1,2,\cdots,\delta_P}$, such that the highest priority job is chosen to be processed.
This priority mapping is simply LIFO according to the time in which the job moved from $F$ to $P$.


\begin{defn}[violation]
    Observe two jobs $q_1,q_2 \in Q_F$ in the algorithm.
    We say that the ordered pair $(q_1,q_2)$ is a \emph{violation} if $\p{F}(q_1) > \p{F}(q_2)$ and $\dstr\cdot \ept{q_2} \le \ept{q_1}$.
    Informally, the algorithm assigns priority to $q_1$ over $q_2$ even though it knows that the processing time of $q_2$ (which is strictly smaller than  $\dstr\cdot \ept{q_2}$) is strictly smaller than that of $q_1$.
\end{defn}

\textbf{Algorithm's Description and Intuition.}
The algorithm maintains two properties regarding the priorities in the bin $F$.
The first property is that there exist no violations.
When this holds, the priority is ``\textbf{quasi-SRPT}''; that is, if a job has a higher priority than another job, then the lower bound on the processing time of the first job is no more than the upper bound on the processing time of the second job (noting that in $F$ the remaining processing time is equal to the initial processing time).
This ``quasi-SRPT'' property is used in \cref{prop:UW_RelevantArrival}.
The second property is that for every natural $k$, the total processing time of the $k$ lowest priority jobs never decreases upon the release of a new job, which is used in  \cref{prop:UW_IrrelevantArrival}.

The way in which the algorithm maintains both of those traits for the bin $F$ can be seen in the function $\UponJobRelease{q}$, which is called upon the release of a new job $q$.
The algorithm first gives the newly-released job the maximum priority in $F$, which does not affect the second property, but could possibly break the first property by causing violations.
In order to fix these violations, the algorithm performs a \textbf{rotation} of the new job $q$ and all violations.
That is, numbering the jobs in violation with $q$ as $q_{1},\cdots, q_m$ by order of decreasing priority, the rotation \underline{simultaneously} sets the new priority of $q$ to be the old priority of $q_m$, the new priority of $q_m$ to be the old priority of $q_{m-1}$, and so on (the new priority of $q_1$ is the old priority of $q$, i.e. the maximum priority).
In \cref{alg:UW_Algorithm}, this rotation operation is described as simply composing the mapping $\p{F}$ with the cyclic permutation $\sigma = \pr{\p{F}(q_m), \p{F}(q_{m-1}),\cdots,\p{F}(q_{1}),\p{F}(q)}$, where $\sigma$ maps from the natural numbers $\pb{\delta_F}$ to $\pb{\delta_F}$, and is written in cycle notation.
In the analysis, we show that this rotation is sufficient to solve any violations that occur in the insertion.

\paragraph{Visualization in Figures.} Visualizations of the state of \cref{alg:UW_Algorithm} are given throughout the paper (e.g. in \cref{fig:UW_Structure}).
In these visualizations, each job is visualized as a rectangle, where the height of the rectangle is the weight of the job (which is currently $1$, since we're considering the unweighted setting).
The area of the rectangle is the remaining volume of the job. Consider the transformation of the rectangle as a job is processed: its weight (height) remains the same but its volume (area) decreases.
Hence, the width of the rectangle decreases as the job is processed.

A bin (either $F$ or $P$) is visualized as a ``stack'' of the jobs in that bin.
Those jobs are stacked according to their priority in the bin (the highest priority job is at the top of the stack).
At any point in time, the highest priority job in $P$ is processed -- this is the job at the top of the visualization of $P$.

These visualizations are for the sake of understanding the algorithm and its proof, and are not used in the algorithm itself.
The algorithm is not aware of the details of the visualization, and in particular is not aware of the volume of the jobs (i.e. the area of the rectangles).

The algorithm is shown in \cref{alg:UW_Algorithm}. \cref{fig:UW_Structure,fig:UW_Insertion} visualize \cref{alg:UW_Algorithm}.

\begin{algorithm}[h]
    \renewcommand{\algorithmcfname}{Algorithm}
    \caption{\label{alg:UW_Algorithm} Scheduling with Predictions -- Unweighted}

    \EFn(\tcp*[h]{upon the release of a job $q$ at time $t$}){\UponJobRelease{$q$}}{
    Set $Q_F \gets Q_F \cup \pc{q}$, and set $\p{F}(q) \gets \delta_F$.

    \BlankLine

    \tcp*[h]{rotate the priorities of jobs to fix violations}

    \If{there exists $q'\in Q_F\backslash\pc{q}$ such that $(q,q')$ is a violation}{\label{line:UW_Rotate}
    let $Q' = \pc{q'\in Q_F\backslash\{q\} \middle| (q,q')\text{ is a violation}}$.

    denoting $m = \ps{Q'}$, let $q_1,\cdots,q_m$ be an enumeration of $Q'$ such that $\p{F}(q) > \p{F}(q_1) > \cdots > \p{F}(q_m)$.

    let $\sigma$ be the cyclic permutation $\pr{\p{F}(q_m), \p{F}(q_{m-1}),\cdots,\p{F}(q_{1}),\p{F}(q)}$.

    set $\p{F} \gets \sigma \circ \p{F}$
    }
        }

    \BlankLine

    \EFn(\tcp*[h]{upon $\delta_F>\delta_P$}){\UponHeavyF{}}{
    let $q \in Q_F$ be such that $\p{F}(q) = \delta_F$.

    move $q$ from $F$ to $P$, and set $\p{P}(q) = \delta_P$.

    }

    \BlankLine

    \EFn(\tcp*[h]{at any point in time $t$}){\Process{}}{
    let $q\in Q_P$ be such that $\p{P}(q) = \delta_P$.

    process $q$.
    }

\end{algorithm}

\begin{figure*}[h]
    \begin{center}
        \hspace*{\fill}
        \subfloat[][\label{subfig:UW_Structure} The state of \cref{alg:UW_Algorithm}]{ \includegraphics[width=0.3\textwidth]{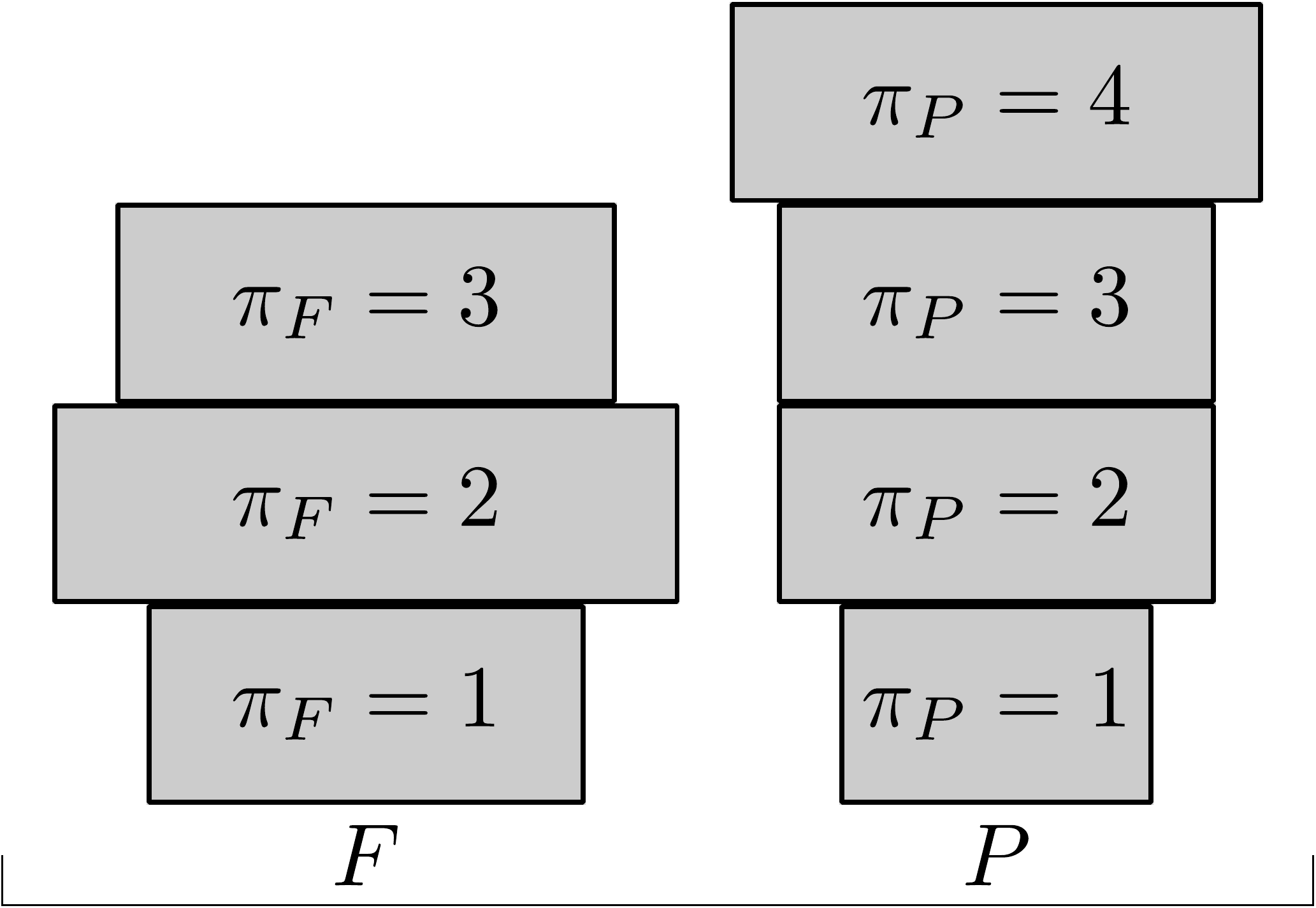}}
        \hspace*{\fill}
        \vline
        \hspace*{\fill}
        \subfloat[][\label{subfig:UW_Transfer1}$ \UponHeavyF $ is called since $\delta_F > \delta_P$.]{\includegraphics[width=0.3\textwidth]{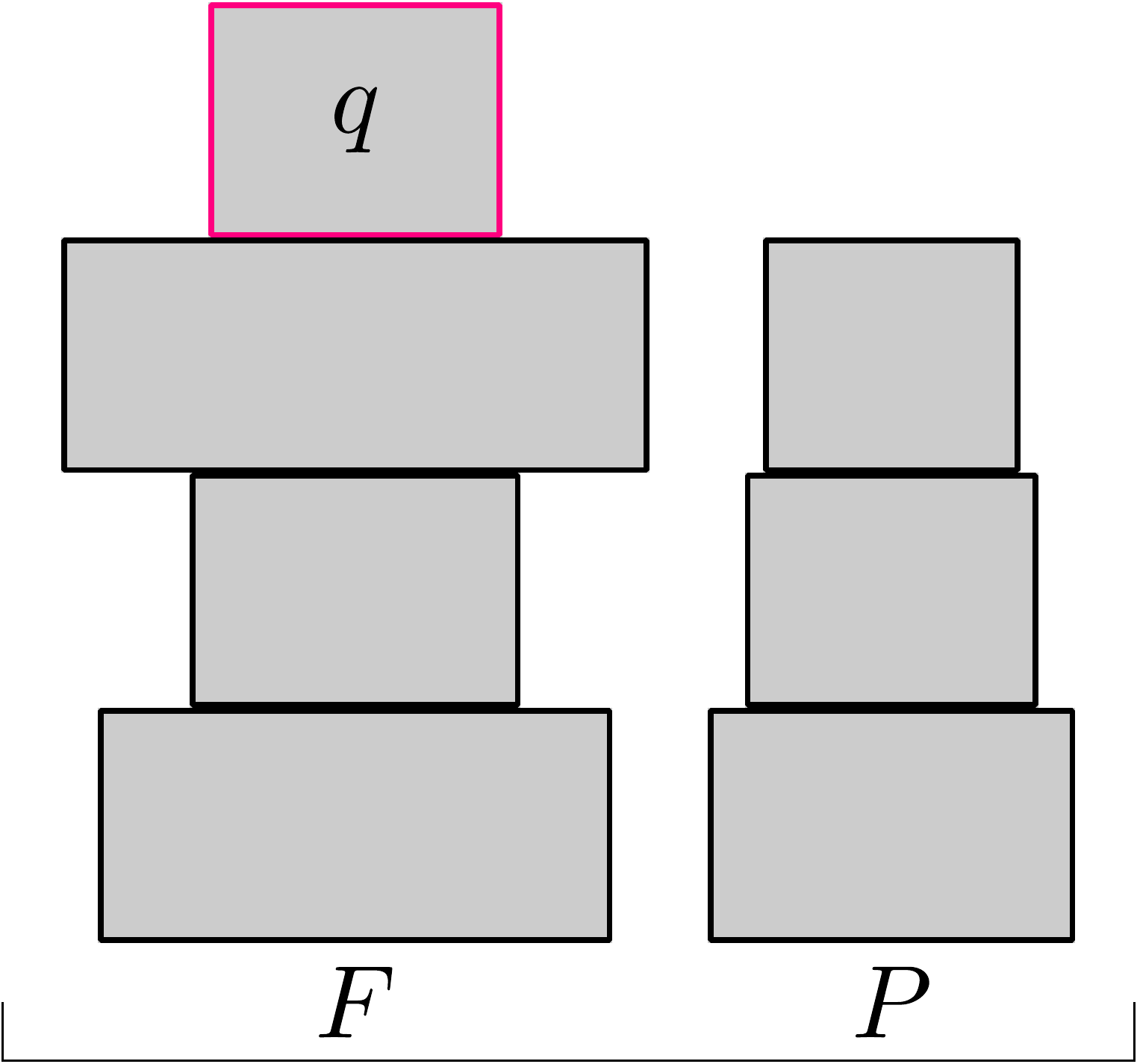}}
        \hspace*{\fill}
        \subfloat[][\label{subfig:UW_Transfer2}Request $q$ receives highest priority, causing violations with $q_1,q_2,q_3$.]{
        \includegraphics[width=0.3\textwidth]{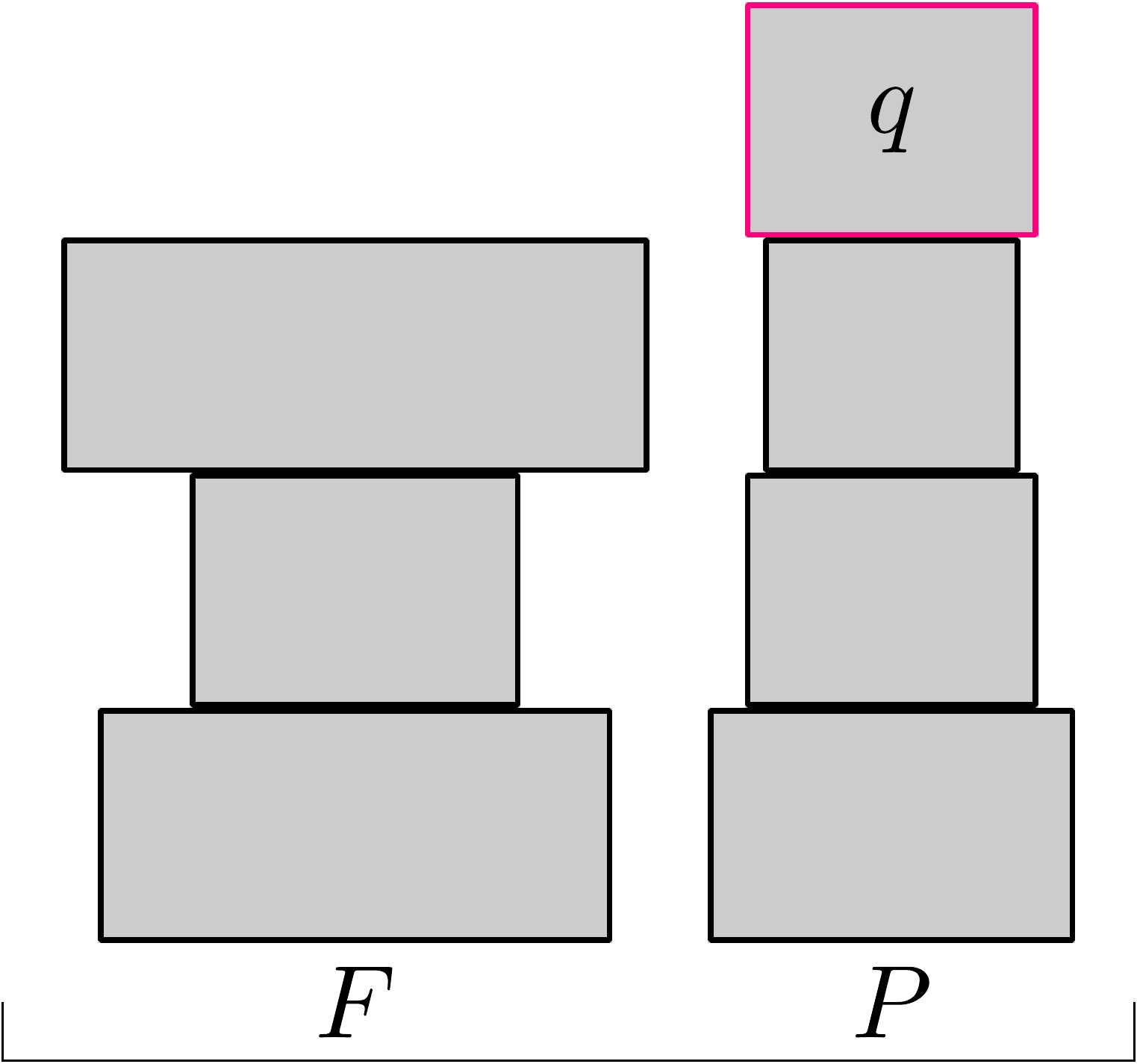}}
        \hspace*{\fill}
    \end{center}
    \Cref{subfig:UW_Structure} shows a possible state of the algorithm at some point in time.
    In the figure, the bin $F$ contains $3$ pending jobs, and the bin $P$ contains $4$ pending jobs.
    In this figure, the job being processed is the job of priority $4$ in $P$.
    \caption[]{\label{fig:UW_Structure}The State of Bins and $\UponHeavyF$ in \cref{alg:UW_Algorithm}}

    \cref{subfig:UW_Transfer1,subfig:UW_Transfer2} show the operation of $\UponHeavyF$.
    At the state of \cref{subfig:UW_Transfer1}, $\delta_F > \delta_P$ and so $\UponHeavyF$ is called.
    The top priority job $q$ in $F$ is then moved to $P$ , where it also has top priority (\cref{subfig:UW_Transfer2}).
\end{figure*}

%
%

\begin{figure*}

    \begin{center}
        \subfloat[][\label{subfig:UW_Insertion1}Request $q$ is released, and sent to bin $F$]{
        \includegraphics[height=0.3\textheight]{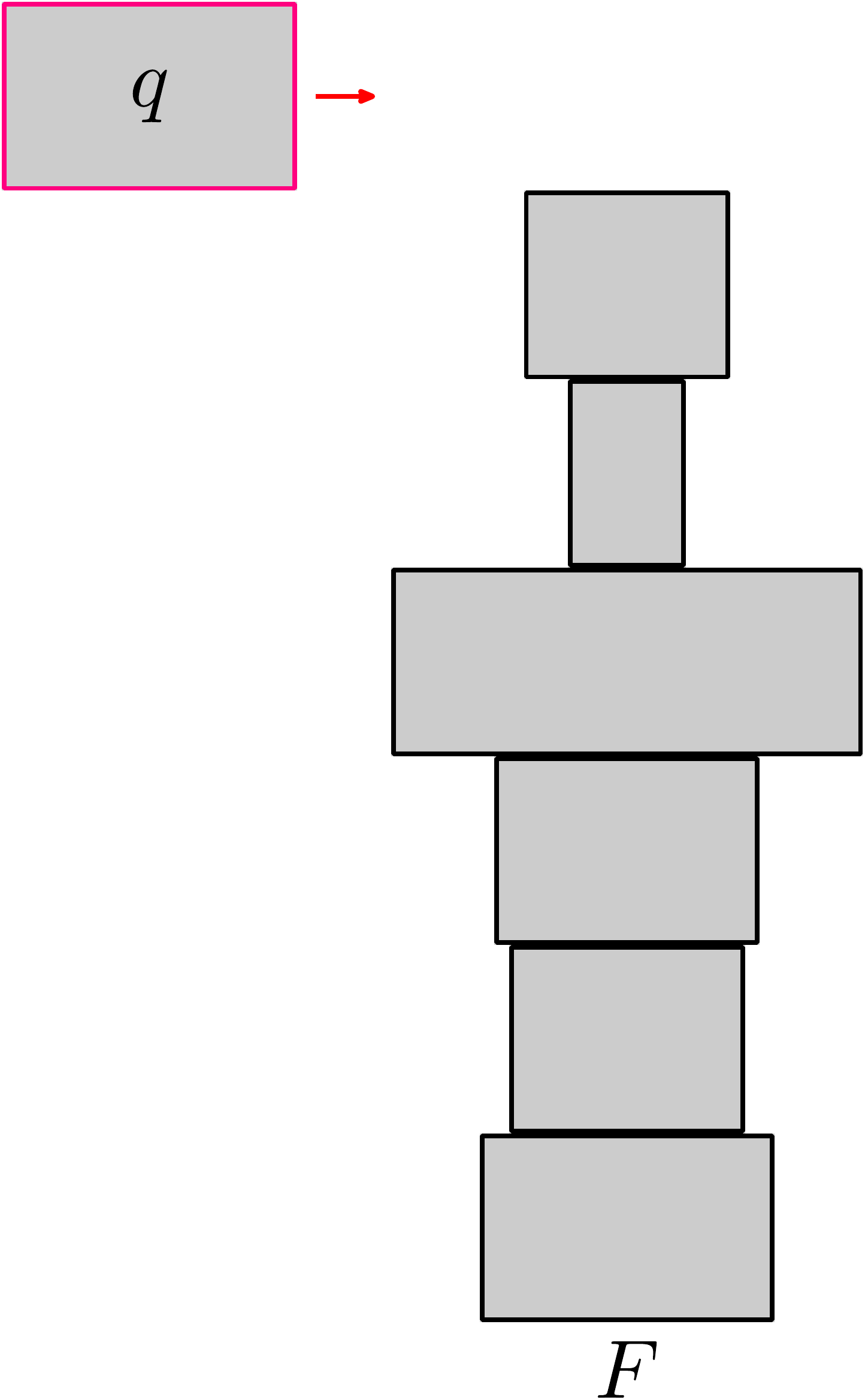}
        }
        \hfill
        \subfloat[][\label{subfig:UW_Insertion2}Request $q$ receives highest priority, causing violations with $q_1,q_2,q_3$.]{
        \hspace{0.1\textwidth}
        \includegraphics[height=0.3\textheight]{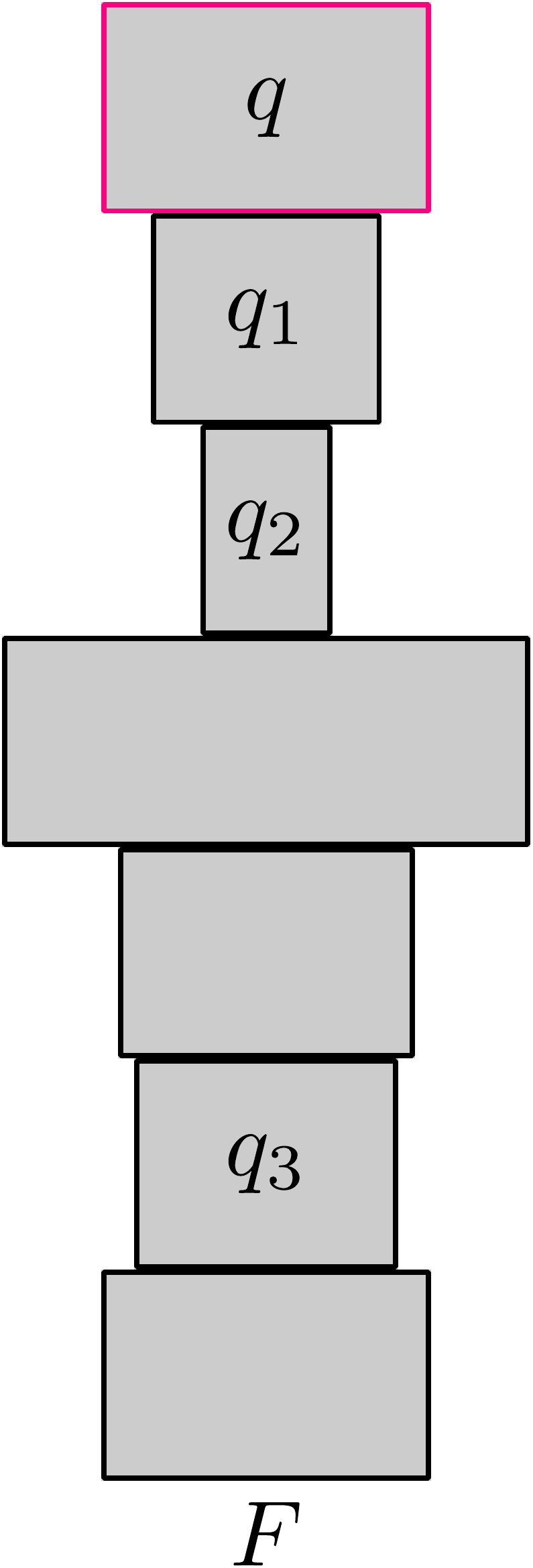}
        \hspace{0.1\textwidth}
        }
        \hfill
        \subfloat[][\label{subfig:UW_Insertion3}A rotation occurs, solving the violations.]{
        \hspace{0.05\textwidth}
        \includegraphics[height=0.3\textheight]{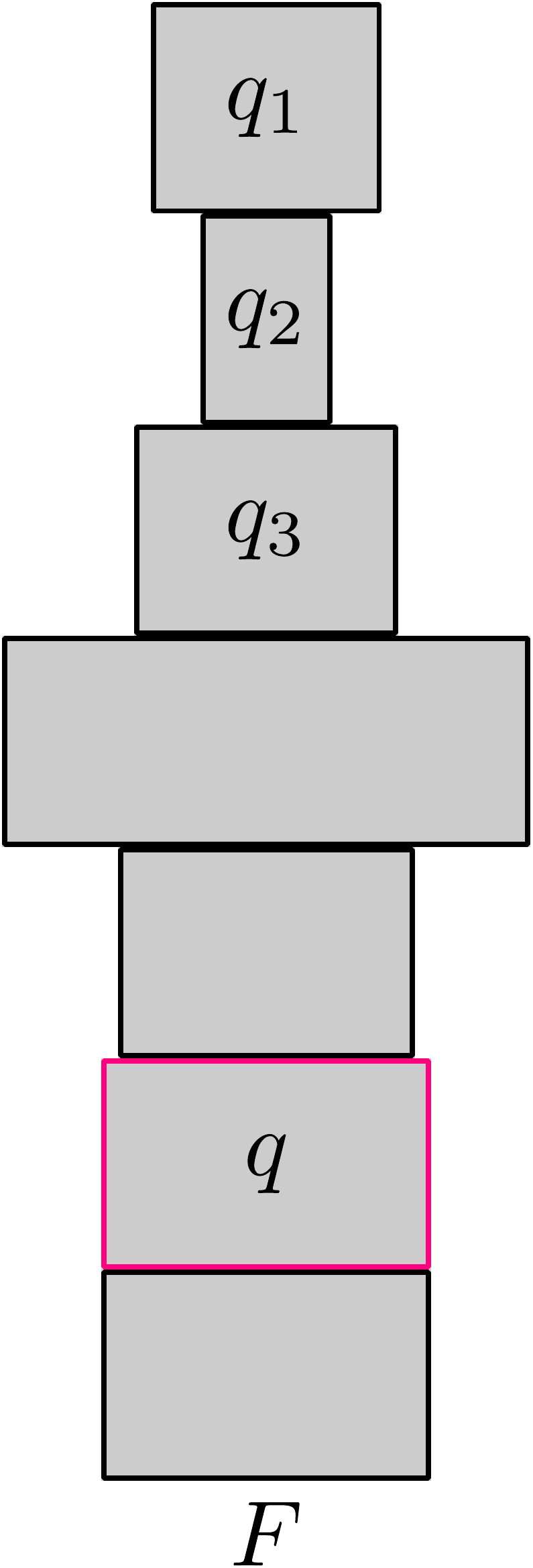}
        \hspace{0.05\textwidth}
        }
    \end{center}

    \caption[]{\label{fig:UW_Insertion}$ \UponJobRelease $ in \cref{alg:UW_Algorithm}}
\end{figure*}

The competitiveness guarantee of \Cref{alg:UW_Algorithm} is given in \cref{thm:UW_Competitiveness}, the proof of which is given in \cref{subsec:LogW_Proofs_UW}.
\begin{thm}
    \label{thm:UW_Competitiveness}
    \Cref{alg:UW_Algorithm} is a $\dstr$-robust, $2\ceil{\dstr^2}$-competitive algorithm for the unweighted \prob problem.
\end{thm}

\subsection{The $O(\dstr^2 \log W)$-competitive Algorithm for Weighted Flow Time}
\label{subsec:LogW}

Having described the unweighted algorithm in the previous subsections, we present the $O(\dstr^2 \log W)$-competitive algorithm for the weighted setting.

\subsection{The Algorithm}
\label{subsec:LogW_Algorithm}

We assume that the weights of incoming jobs are of an integer power of $2$; the algorithm can enforce this by rounding weights, incurring a factor of at most $2$ to its competitive ratio.

For each $i\in \mathbb{Z}$, the algorithm maintains a \emph{superbin} $A^i$ (only nonempty superbins are maintained explicitly).
The superbin is constructed as in \cref{subsec:UW}.
That is, $A^i$ contains two bins $F^i$ and $P^i$ (containing the pending jobs $Q_{F^i}$ and $Q_{P^i}$, respectively), and maintains the priority bijections $\p{F^i}$ and $\p{P^i}$ on the jobs of $F^i$ and $P^i$, respectively.
As in the weighted case, the algorithm refers to $\ps{Q_{F^i}}$ and $\ps{Q_{P^i}}$ by $\delta_{F^i}$ and $\delta_{P^i}$, respectively.

The algorithm sends released jobs of weight $2^i$ (for any $i\in \mathbb{Z}$) to the superbin $A^i$.
The insertion to the superbin $A^i$ is according to $\UponJobRelease$ in \cref{alg:UW_Algorithm}.
That is, the released job $q$ is first inserted to $F^i$, the algorithm sets $\p{F^i}(q) \gets \delta_{F^i}$, and then possibly performs a rotation to solve violations.
For every superbin index $i$, the algorithm also maintains that $\wts{Q_{F^i}}\le \wts{Q_{P^i}}$, by moving jobs from $F^i$ to $P^i$ as in \cref{subsec:UW}, performed by calling $\UponHeavyF$ as defined in \cref{alg:UW_Algorithm}.

As for processing, the algorithm chooses for processing the superbin $A^i$ with the heaviest partial bin $P^i$.
When a superbin $A^i$ is chosen for processing, algorithm makes a call to $\Process$ as defined in \cref{alg:UW_Algorithm}, which chooses the job $q\in Q_{P^i}$ such that $\p{P^i}(q)$ is maximal.

The algorithm is given in \cref{alg:LogW_Algorithm}.
A possible state of \cref{alg:LogW_Algorithm} is visualized in \cref{fig:LogW_Structure}.

\begin{algorithm}[h]
    \caption{\label{alg:LogW_Algorithm} Scheduling with Predictions -- $O(\log W)$-Competitive}

    \EFn(\tcp*[h]{upon the release of a job $q$ at time $t$}){\UponJobRelease{$q$}}{

        let $i$ be such that $\wt{q} = 2^i$.

        call $A^i.\UponJobRelease{q}$
    }

    \BlankLine

    \EFn(\tcp*[h]{at any point in time $t$}){\Process{}}{
        let $i = \arg\max_{i'} \wts{P^{i'}}$.

        call $A^i.\Process{}$.
    }

    \EFn(\tcp*[h]{upon $w(F^i)>w(P^i)$ for some index $i$}){\UponHeavyF{$i$}}{
        call $A^i.\UponHeavyF{}$
    }

\end{algorithm}

\begin{figure}
    \begin{center}
        \includegraphics[height=0.3\textheight]{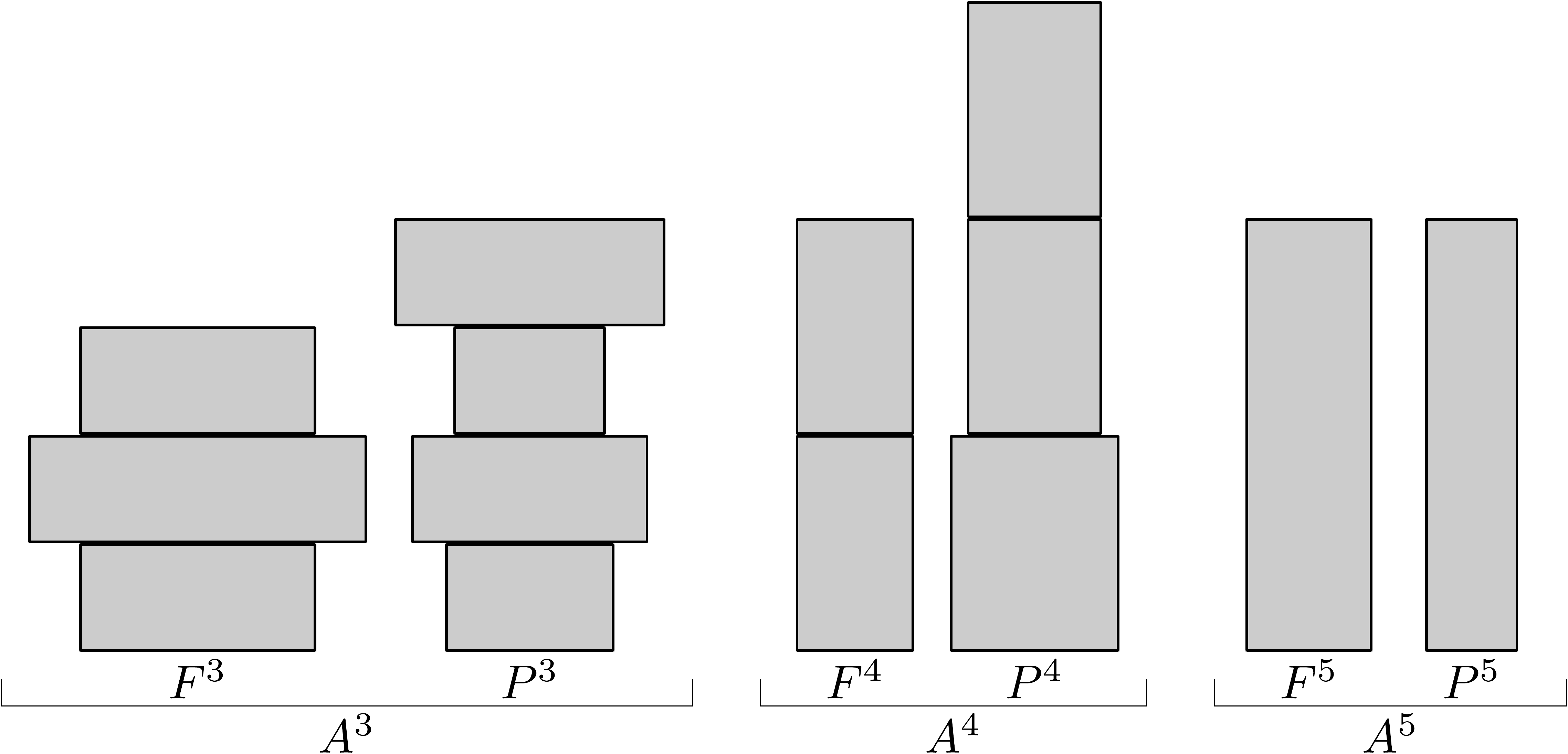}
    \end{center}
    \caption[]{A Possible State of \cref{alg:LogW_Algorithm}}
    \label{fig:LogW_Structure}
\end{figure}

The competitiveness guarantee of \cref{alg:LogW_Algorithm} is given in \cref{thm:LogW_Competitiveness}, the proof of which appears in \cref{subsec:LogW_Proofs_W}.

\begin{thm}
    \label{thm:LogW_Competitiveness}
    \cref{alg:LogW_Algorithm} is a $\dstr$-robust, $O\pr{\dstr^2\log W}$-competitive algorithm for the \prob problem.
\end{thm}


\section{Discussion and Open Problems}
\label{sec:Disc}

This paper presents algorithms for {\prob} for minimizing weighted flow time which have a polynomial dependency on the distortion $\dstr$ of the given predictions, such that when this distortion is constant the best known competitiveness bounds are matched (namely, $O(\log P)$, $O(\log D)$ and $O(\log W)$).
These are the first algorithms to maintain \emph{any} nontrivial competitiveness when given inaccurate processing times in the weighted flow time setting.

However, in the model which we consider, each such $\dstr$-robust algorithm is tailored to a specific value of $\dstr$.
One could desire an algorithm which works for \emph{any} $\dstr$, while maintaining a competitive ratio which is a function of this $\dstr$.
In many other problems, going from the first model to the second model can be very easily done through a doubling procedure on the parameter $\dstr$; however, applying such a doubling scheme to this scheduling problem does not seem immediate, and we leave its development to future work.

The dependencies of our algorithms on the maximum distortion are $O(\dstr^2)$ and $\tilde{O}(\dstr^3)$.
Improving the dependence on $\dstr$ in the competitiveness is also an interesting problem; we conjecture that a linear dependence on $\dstr$ is the optimal one.

Another problem is determining the exact competitive ratio for $\dstr$-robust algorithm as $\dstr$ approaches $1$ (i.e. small distortion).
This problem is most salient for unweighted flow time, in which we've shown a lower bound of $2$-competitiveness, and our upper bound result only yields $4$-competitiveness.

\bibliographystyle{plain}
\bibliography{bibfile}

\begin{thebibliography}{10}

\bibitem{DBLP:conf/focs/AzarT18}
Yossi Azar and Noam Touitou.
\newblock Improved online algorithm for weighted flow time.
\newblock In Mikkel Thorup, editor, {\em 59th {IEEE} Annual Symposium on
  Foundations of Computer Science, {FOCS} 2018, Paris, France, October 7-9,
  2018}, pages 427--437. {IEEE} Computer Society, 2018.

\bibitem{DBLP:conf/soda/BansalC09}
Nikhil Bansal and Ho{-}Leung Chan.
\newblock Weighted flow time does not admit o(1)-competitive algorithms.
\newblock In {\em Proceedings of the Twentieth Annual {ACM-SIAM} Symposium on
  Discrete Algorithms, {SODA} 2009, New York, NY, USA, January 4-6, 2009},
  pages 1238--1244, 2009.

\bibitem{DBLP:journals/talg/BansalD07}
Nikhil Bansal and Kedar Dhamdhere.
\newblock Minimizing weighted flow time.
\newblock {\em {ACM} Trans. Algorithms}, 3(4):39, 2007.
\newblock also in SODA 2003: 508-516.

\bibitem{DBLP:journals/algorithmica/BansalDKS04}
Nikhil Bansal, Kedar Dhamdhere, Jochen K{\"{o}}nemann, and Amitabh Sinha.
\newblock Non-clairvoyant scheduling for minimizing mean slowdown.
\newblock {\em Algorithmica}, 40(4):305--318, 2004.

\bibitem{Becchetti2001}
Luca Becchetti and Stefano Leonardi.
\newblock Non-clairvoyant scheduling to minimize the average flow time on
  single and parallel machines.
\newblock In {\em Proceedings of the thirty-third annual ACM symposium on
  Theory of computing}, pages 94--103, 2001.

\bibitem{DBLP:journals/tcs/BecchettiLMP04}
Luca Becchetti, Stefano Leonardi, Alberto Marchetti{-}Spaccamela, and Kirk
  Pruhs.
\newblock Semi-clairvoyant scheduling.
\newblock {\em Theor. Comput. Sci.}, 324(2-3):325--335, 2004.

\bibitem{DBLP:conf/soda/BenderMR02}
Michael~A. Bender, S.~Muthukrishnan, and Rajmohan Rajaraman.
\newblock Improved algorithms for stretch scheduling.
\newblock In David Eppstein, editor, {\em Proceedings of the Thirteenth Annual
  {ACM-SIAM} Symposium on Discrete Algorithms, January 6-8, 2002, San
  Francisco, CA, {USA}}, pages 762--771. {ACM/SIAM}, 2002.

\bibitem{DBLP:conf/stoc/ChekuriKZ01}
Chandra Chekuri, Sanjeev Khanna, and An~Zhu.
\newblock Algorithms for minimizing weighted flow time.
\newblock In {\em Proceedings on 33rd Annual {ACM} Symposium on Theory of
  Computing, July 6-8, 2001, Heraklion, Crete, Greece}, pages 84--93, 2001.

\bibitem{Im2017}
Sungjin Im, Janardhan Kulkarni, and Kamesh Munagala.
\newblock Competitive algorithms from competitive equilibria: Non-clairvoyant
  scheduling under polyhedral constraints.
\newblock {\em Journal of the ACM (JACM)}, 65(1):1--33, 2017.

\bibitem{Im2014}
Sungjin Im, Janardhan Kulkarni, Kamesh Munagala, and Kirk Pruhs.
\newblock Selfishmigrate: A scalable algorithm for non-clairvoyantly scheduling
  heterogeneous processors.
\newblock In {\em 2014 IEEE 55th Annual Symposium on Foundations of Computer
  Science}, pages 531--540. IEEE, 2014.

\bibitem{DBLP:journals/ipl/KimC03a}
Jae{-}Hoon Kim and Kyung{-}Yong Chwa.
\newblock Non-clairvoyant scheduling for weighted flow time.
\newblock {\em Inf. Process. Lett.}, 87(1):31--37, 2003.

\bibitem{LattanziLMV20}
Silvio Lattanzi, Thomas Lavastida, Benjamin Moseley, and Sergei Vassilvitskii.
\newblock Online scheduling via learned weights.
\newblock In {\em Proceedings of the Thirty-First Annual {ACM-SIAM} Symposium
  on Discrete Algorithms, {SODA} 2020, New Orleans, LA, USA, January 5 - 8,
  2020.}, 2020.

\bibitem{DBLP:conf/icml/LykourisV18}
Thodoris Lykouris and Sergei Vassilvitskii.
\newblock Competitive caching with machine learned advice.
\newblock In {\em Proceedings of the 35th International Conference on Machine
  Learning, {ICML} 2018, Stockholmsm{\"{a}}ssan, Stockholm, Sweden, July 10-15,
  2018}, pages 3302--3311, 2018.

\bibitem{DBLP:conf/nips/MedinaV17}
Andres~Mu{\~{n}}oz Medina and Sergei Vassilvitskii.
\newblock Revenue optimization with approximate bid predictions.
\newblock In {\em Advances in Neural Information Processing Systems 30: Annual
  Conference on Neural Information Processing Systems 2017, 4-9 December 2017,
  Long Beach, CA, {USA}}, pages 1858--1866, 2017.

\bibitem{DBLP:journals/corr/abs-2006-15463}
Michael Mitzenmacher.
\newblock Queues with small advice.
\newblock {\em CoRR}, abs/2006.15463, 2020.

\bibitem{mitzenmacher:LIPIcs:2020:11699}
Michael Mitzenmacher.
\newblock {Scheduling with Predictions and the Price of Misprediction}.
\newblock In Thomas Vidick, editor, {\em 11th Innovations in Theoretical
  Computer Science Conference (ITCS 2020)}, volume 151 of {\em Leibniz
  International Proceedings in Informatics (LIPIcs)}, pages 14:1--14:18,
  Dagstuhl, Germany, 2020. Schloss Dagstuhl--Leibniz-Zentrum fuer Informatik.

\bibitem{mitzenmacher2020algorithms}
Michael Mitzenmacher and Sergei Vassilvitskii.
\newblock Algorithms with predictions.
\newblock {\em arXiv preprint arXiv:2006.09123}, 2020.

\bibitem{Motwani1994}
Rajeev Motwani, Steven Phillips, and Eric Torng.
\newblock Nonclairvoyant scheduling.
\newblock {\em Theoretical computer science}, 130(1):17--47, 1994.

\bibitem{purohit2018improving}
Manish Purohit, Zoya Svitkina, and Ravi Kumar.
\newblock Improving online algorithms via ml predictions.
\newblock In {\em Advances in Neural Information Processing Systems}, pages
  9661--9670, 2018.

\bibitem{NAV:NAV3800030106}
Wayne~E. Smith.
\newblock Various optimizers for single-stage production.
\newblock {\em Naval Research Logistics Quarterly}, 3(1-2):59--66, 1956.

\bibitem{DBLP:journals/mst/ZhuCL15}
Jianqiao Zhu, Ho{-}Leung Chan, and Tak~Wah Lam.
\newblock Non-clairvoyant weighted flow time scheduling on different
  multi-processor models.
\newblock {\em Theory Comput. Syst.}, 56(1):82--95, 2015.

\end{thebibliography}

\newpage

\appendix

\section{Additional proofs from \cref{sec:CLP}}
\label{sec:CLP_Proofs}

\begin{proof}[Proof of \cref{lem:CLP_VolumeConversion}]
    Observe that
    $
        \sum_{j=j_1}^{j_2} \Wtp{\le i}{=j}{t}
        = \sum_{j=j_1}^{j_2} \WtpP{\le i}{=j}{t} + \sum_{j=j_1}^{j_2} \WtpF{\le i}{=j}{t}
    $,
    where the superscripts $\textrm{p}$ and $\textrm{f}$ refer to partial and full jobs, respectively.

    The algorithm maintains that in every weight class + \id class combination, there exists at most one partial job of those classes.
    Thus, for every \id class $j$, a geometric sum yields that  $\WtpP{\le i}{=j}{t} \le 2w^j$, where $w^j$ is the maximum weight of a partial job in $\reqp{\le i}{=j}{t}$.
    Since $w^j \le \qw{i}$, we thus have that $\WtpP{\le i}{=j}{t} \le 2\qw{i}$, and thus $\sum_{j=j_1}^{j_2} \WtpP{\le i}{=j}{t} \le 2(j_2-j_1+1) \qw{i} $.
    In addition, each weight class can exist in at most $\nclass$ different \id classes, which yields

    \[
        \sum_{j=j_1}^{j_2} \WtpP{\le i}{=j}{t} \le \sum_{j=j_1}^{j_2} 2w^j \le 2\qw{i}\nclass +  2\qw{i-1}\nclass+\cdots \le 4 \qw{i}\nclass
    \]

    Combining, we have $\sum_{j=j_1}^{j_2} \WtpP{\le i}{=j}{t} \le \min\pc{2(j_2-j_1+1), 4\nclass}\cdot \qw{i}$.

    It remains to bound the weight of the full jobs, namely $\sum_{j=j_1}^{j_2} \WtpF{\le i}{=j}{t}$.
    \begin{align*}
        \sum_{j=j_1}^{j_2} \WtpF{\le i}{=j}{t}
        &\le \sum_{j=j_1}^{j_2} \frac{\Vp{\le i}{=j}{t}}{2^{j}}\\
        &\le \sum_{j=j_1}^{j_2} \frac{\Delta\Vp{\le i}{=j}{t} + \Vp*{\le i}{=j}{t}}{2^{j}} \\
        &\le 2\dstr\sum_{j=j_1}^{j_2}\Wtp*{\le i}{=j}{t} + \sum_{j=j_1}^{j_2} \frac{\Delta\Vp{\le i}{\le j}{t} - \Delta\Vp{\le i}{\le j-1}{t}}{2^{j}}\\
        &\le 2\dstr\sum_{j=j_1}^{j_2}\Wtp*{\le i}{=j}{t} +\frac{\Delta\Vp{\le i}{\le j_2}{t}}{2^{j_2}} +\sum_{j=j_1}^{j_2-1} \frac{\Delta\Vp{\le i}{\le j}{t}}{2^{j+1}} -\frac{\Delta\Vp{\le i}{\le j_1-1}{t}}{2^{j_1}} \\
        &\le 2\dstr\sum_{j=j_1}^{j_2}\Wtp*{\le i}{=j}{t} +\frac{\Delta\Vp{\le i}{\le j_2}{t}}{2^{j_2}} +\sum_{j=j_1}^{j_2-1} \frac{\Delta\Vp{\le i}{\le j}{t}}{2^{j+1}} + \dstr\Wtp*{\le i}{\le j_1-1}{t} \\
        &\le 2\dstr\sum_{j=j_1}^{j_2}\Wtp*{\le i}{=j}{t} +\sum_{j=j_1}^{j_2} \max\pc{0, \frac{\Delta\Vp{\le i}{\le j}{t}}{2^{j}}} + \dstr\Wtp*{\le i}{\le j_1-1}{t} \\
        &\le 2\dstr\Wtp*{\le i}{\le j_2}{t} +\sum_{j=j_1}^{j_2} \max\pc{0, \frac{\Delta\Vp{\le i}{\le j}{t}}{2^{j}}}\\
        &\le 2\dstr\Wt*{t} +\sum_{j=j_1}^{j_2} \max\pc{0, \frac{\Delta\Vp{\le i}{\le j}{t}}{2^{j}}}
    \end{align*}
    where the first inequality is due to the fact that a full job $q$ of \id class $j$ has at least $\wt{q} \cdot 2^j$ remaining volume, the second and fifth inequalities are due to the fact that every job $q$ of \id class at most $j$ in the optimal solution has at most $2^{j+1}\dstr \wt{q}$ volume.
\end{proof}

\begin{proof}[Proof of \cref{prop:CLP_DeltaAndWeightBounds}]
    First, we show that for every $j$, the first claim implies the second claim.
    Then, we prove the first claim by induction on the \id classes, in decreasing order.
    Combining these two proofs, the proposition holds.

    \textbf{The first claim implies the second claim.}
    Assume that $\clw{\lastt{j}} \le \owc$ for some $j$.
    Since the optimal solution has no jobs of weight more than $\qw{\owc}$ alive at time $t$, and since the algorithm did have any such jobs alive at $\lastt{j}$, it must be that the algorithm worked on such jobs during $\IR{\lastt{j}}{t}$ at least as much as the optimal solution.
    In addition, the algorithm did not work on any jobs of weight at most $\qw{\owc}$ and \id class more than $j$ during this interval.
    Thus, it holds that
    \[
        \Delta \Vp{\le \owc}{\le j}{t} \le \Delta \Vp{\le \owc}{\le j}{\lastt{j}} \le   \Vp{\le \owc}{\le j}{\lastt{j}}
    \]

    Thus, it's enough to bound $\Vp{\le \owc}{\le j}{\lastt{j}}$.
    Observe that since the algorithm chose to process a job of \id class more than $j$, the total weight in each \id class $j' \le j$ is at most $\clw{\lastt{j}}$.
    Thus, the total volume at $\lastt{j}$ in \id class $j'$ is at most $\dstr\cdot 2^{j'+1}\qw{\clw{\lastt{j}}}$.
    Summing over all $j' \le j$ yields a geometric sum, which is at most $2^{j+2}\qw{\clw{\lastt{j}}}$, showing that the first claim indeed implies the second claim.

    \textbf{The first claim holds.} We prove the first claim of the proposition by induction on the \id classes, in decreasing order.
    For the base case of the largest \id class $j_{\max}$, we have $\lastt{j_{\max}}=0$ and thus the claim holds.

    Now, let $j$ be any \id class and suppose that the claim holds for all \id classes $j+1$ and above.
    Assume for contradiction that $\clw{\lastt{j}} > \owc$.
    Let $q$ be the job chosen for processing at $\lastt{j}$, and define  $h:= \edc{q}$.
    From the definition of $\lastt{j}$, we know that $\wc{q} \le \owc$ and that $h>j$.

    Since the algorithm always processes the maximum-weight job in the chosen \id class, we know that
    \[
        \Wtp{\le \wc{q}}{=h}{\lastt{j}} = \Wtp{\top}{=h}{\lastt{j}} \ge \qw{\clw{\lastt{j}}}
    \]
    where the inequality is due to the choice of the algorithm.

    Now, since during $\IR{\lastt{j}}{t}$ the algorithm did not process any jobs of weight class at most $\owc$ and \id class more than $j$, it must be that
    \[
        \Wtp{\le \owc}{=h}{t} \ge \Wtp{\le \owc}{=h}{\lastt{j}} = \Wtp{\le \wc{q}}{=h}{t} \ge \qw{\clw{\lastt{j}}} - \qw{\owc}
    \]
    where the minus term is due to the fact that the algorithm could have completed $q$ exactly at $\lastt{j}$.

    Now, using \cref{lem:CLP_VolumeConversion}, we have
    \begin{align*}
        \Wtp{\le \owc}{=h}{t}
        &\le 2\cdot \qw{\owc} + 2\dstr\Wt*{t} + 2\max\pc{0, \frac{\Delta \Vp{\le \owc}{\le h}{t}[]}{2^h}} \\
        &\le 2\cdot \qw{\owc} + 2\dstr\Wt*{t} + 8\dstr\qw{\clw{\lastt{h}}}   \\
        &\le 2\cdot \qw{\owc} + 2\dstr\Wt*{t} + 8\dstr\qw{\owc}   \\
        &= (8\dstr+ 2)\qw{\owc} + 2\dstr\Wt*{t}
    \end{align*}
    where the second inequality uses the induction hypothesis that the first claim holds for $h$, combined with the previous proof that the first claim implies the second claim.
    The third inequality also uses the induction hypothesis.

    Combining, we have $2\dstr\Wt*{t} \ge \qw{\clw{\lastt{j}}} - (8\dstr +3)\qw{\owc} \ge \qw{\owc+1} - (8\dstr +3)\qw{\owc} = \frac{\bround}{2}\cdot \qw{\owc}$.
    Thus, we get that $\Wt*{t} \ge \frac{\bround}{4\dstr} \cdot \qw{\owc}$, in contradiction to the definition of $\owc$.
    This proves the first claim of the proposition.
\end{proof}

\section{Additional proofs from \cref{sec:LogW}}
\label{sec:LogW_Proofs}
\subsection{Proof of \cref{thm:UW_Competitiveness}}
\label{subsec:LogW_Proofs_UW}
In the following analysis, we prove \cref{thm:UW_Competitiveness}.

\begin{defn}
    For every time $t$, we define $Q_F(t)$ and $Q_P(t)$ to be the values of the variables $Q_F$ and $Q_P$ at time $t$.
    We also define $Q(t)$ to be the set of jobs pending in the algorithm at time $t$, where it holds that $Q(t)=Q_F(t)\cup Q_P(t)$.

    In addition, we use $\delta(t),\delta_F(t), \delta_P(t)$ to denote the cardinalities of $Q(t),Q_F(t),Q_P(t)$ respectively.

    We also use $\p[t]{F}$ and $\p[t]{P}$ to denote the values of $\p{F}$ and $\p{P}$ at time $t$, respectively.
\end{defn}

\begin{defn}
    For every time $t$ and job $q\in Q_F(t)$, define $\base{q}{t} =\pc{ q' \in Q_F(t) \middle| \p{F}(q') \le \p{F}(q)}$.
    For ease of notation, we denote $\wts{\base{q}{t}}$ by $\wbase{q}{t}$.
    We also similarly define $\base{q}{t}$ and $\wbase{q}{t}$ for $q\in Q_P(t)$ according to $\p{P}$.
\end{defn}

\begin{defn}[volume of job covered by bar]
    For any $\xbar{x}$ and time $t$, and for any job $q\in Q(t)$, we define
    \[
        \cont{q}{x}{t} =
        \begin{cases}
            \rpt{q}{t} & x\ge \wbase{q}{t}\\
            0 & \text{otherwise}
        \end{cases}
    \]
    called the volume of $q$ covered by $x$ at $t$.
\end{defn}

\begin{defn}[volume covered by bar]
    We define:
    \begin{enumerate}
        \item 	The covered volume of $x$ at $t$, which is $\cov{x}{t} = \sum_{q\in Q(t)} \cont{q}{x}{t}$.

        \item 	$\cov[F]{x}{t} = \sum_{q\in Q_F(t)} \cont{q}{x}{t}$.

        \item 	$\cov[P]{x}{t} = \sum_{q\in Q_P(t)} \cont{q}{x}{t}$.
    \end{enumerate}
    Observe that $\cov{x}{t} = \cov[F]{x}{t} + \cov[P]{x}{t}$
\end{defn}

\begin{figure}
    \begin{center}
        \includegraphics[height=0.25\textheight]{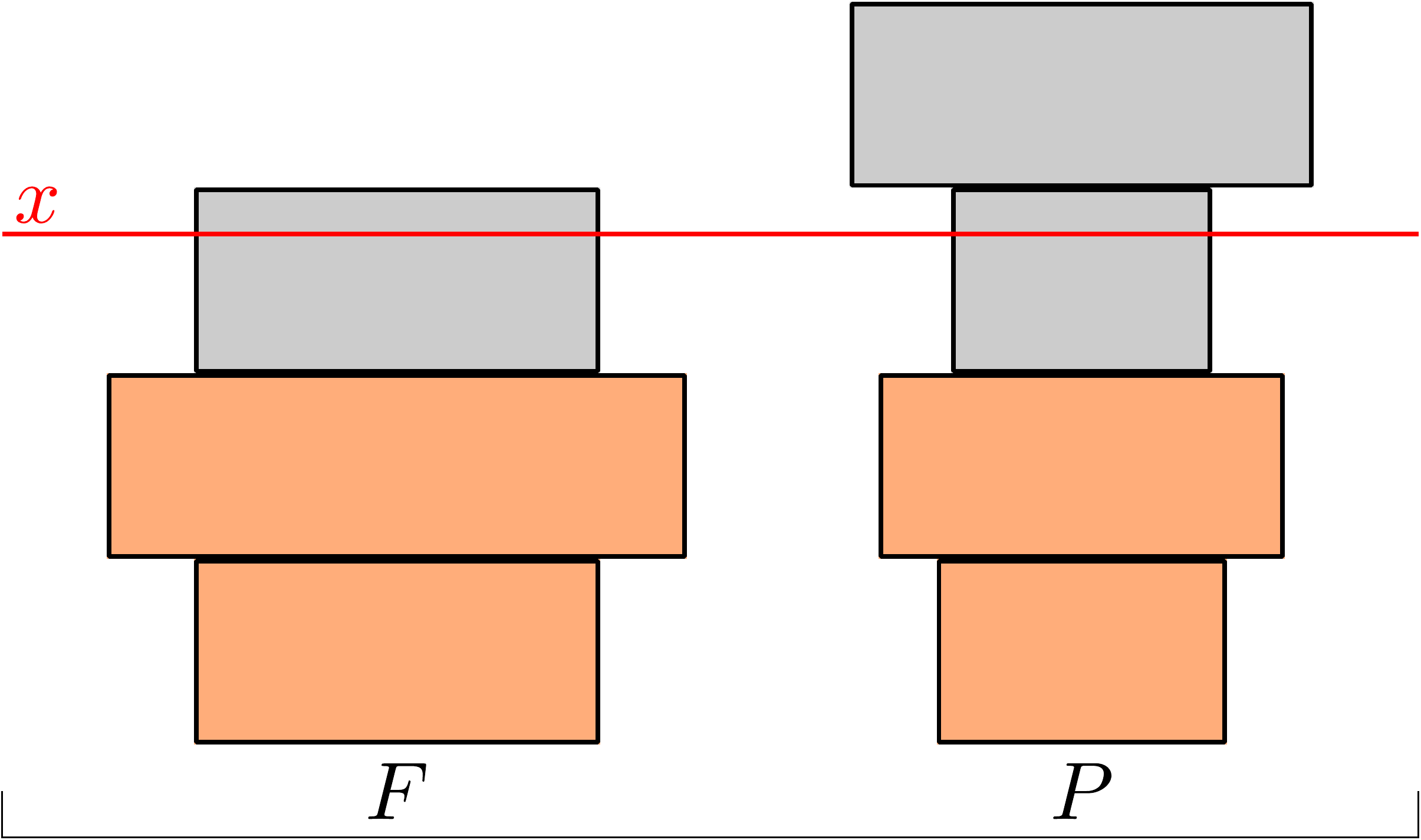}
    \end{center}
    This figure shows the state of the algorithm at some time $t$, and visualizes the volume covered by a bar $\xbar{x}$ at $t$.
    Here, for each job $q$ we have that $\cont{q}{x}{t}$ is the orange-colored area of the job (which is either $0$ or $\rpt{q}{t}$).
    Thus, the total covered volume $\cov{x}{t}$ is the total orange-colored area in the figure.
    \caption{Volume Covered by Bar}
\end{figure}

In the following analysis we consider various properties of the functions of our algorithm.
When considering a function call at time $t$, we use $t^-$ to denote the time immediately before the function call, and use $t$ to denote the time immediately after the function call.

\begin{prop}
    \label{prop:UW_JobMoveRetainsVolume}
    Let $t$ be any time in which $\UponHeavyF$ is called.
    Then for every $\xbar{x}$ we have that
    \[ \cov{x}{t} = \cov{x}{t^-} \]
\end{prop}

\begin{proof}
    Let $q$ be the job moved from $F$ to $P$.
    It holds that $\p[t^-]{F}(q) = \delta_{F}(t^-)$ and that $\p[t]{P}(q)=\delta_{P}(t)$.
    Thus, we have that for every job $q'\in Q(t)\backslash \{q\}$ it holds that $\wbase{q'}{t} = \wbase{q'}{t^-}$, and thus an identical amount of the volume of $q'$ is covered by $x$ at $t^-$ and at $t$.

    As for the volume of $q$, it holds that $\wbase{q}{t^-} = \delta_F(t^-)$.
    Now observe that $\delta_F(t^-) = \delta_P(t^-)+1$, since $\UponHeavyF$ was called (it cannot be that $\delta_F(t^-) > \delta_P(t^-)+1$, since an earlier call to $\UponHeavyF$ would fix this imbalance).
    Thus, after moving $q$ we have that $\delta_P(t) = \delta_F(t^-)$.
    We therefore have that
    \[ \wbase{q}{t} = \delta_P(t)=\delta_F(t^-) = \wbase{q}{t^-} \]

    Overall, we have that $\wbase{q'}{t^-} = \wbase{q'}{t}$ for every $q' \in Q(t)$, and thus $\cov{x}{t} = \cov{x}{t^-}$, as required.
\end{proof}

\begin{prop}
    \label{prop:UW_IrrelevantArrival}
    Consider a call to $\UponJobRelease{q}$ at time $t$.
    For every $\xbar{x}$, it holds that
    \[ \cov{x}{t} \ge \cov{x}{t^-} \]
\end{prop}

\begin{proof}
    We denote by $t_0$ the ``time'' immediately after adding $q$ to $\p{F}$ -- that is, the state of the algorithm before \cref{line:UW_Rotate} in $\UponJobRelease(q)$.


    \textbf{Claim 1: $\cov{x}{t_0} \ge \cov{x}{t^-}$.} Since $\p[t_0]{F}(q) = \delta_F(t_0)$, we have that $\p[t_0]{F}(q) > \p[t_0]{F}(q')$ for every $q'\in Q_F(t^-)$.
    For every such $q'$ we thus have $\wbase{q'}{t_0} = \wbase{q'}{t^-}$, and therefore $\cov{x}{t_0} \ge \cov{x}{t^-}$.

    \textbf{Claim 2: $\cov{x}{t} \ge \cov{x}{t_0}$.} If no violations exist at $t_0$, and thus no rotation is performed, then the claim holds.
    Otherwise, the cyclic permutation $\pr{\p{F}(q_m), \p{F}(q_{m-1}),\cdots,\p{F}(q_{1}),\p{F}(q)}$ is applied, where the jobs $q_1,\cdots,q_m$ are the jobs involved in violations with $q$ at $t_0$, such that $\p[t_0]{F}(q_i) > \p[t_0]{F}(q_{i-1})$ for every $i>0$.

    We write $q_0 = q$, and observe the set of jobs $Q' = \pc{q_0,q_1,\cdots q_m}$ involved in the cyclic permutation.
    Since the only change is in priorities between $t$ and $t_0$ is in $Q'$, it remains to observe the change in the volume of $Q'$ covered by $x$.
    If $\p[t_0]{F}(q') > x$ for every $q'\in Q'$, the volume of $Q'$ covered by $x$ does not change.
    The same holds if $\p[t_0]{F}(q') \le x$ for every $q' \in Q'$.

    It remains to observe the case in which $x$ covers some of $Q'$ but not all.
    In this case, let $i$ be the smallest index such that $\p[t_0]{F}(q_i) \le x$, such that $i>0$.
    After applying the cyclic permutation, we have that the volume covered by $x$ has decreased by $\pt{q_i}$ (since $\p[t]{F}(q_i) = \p[t_0]{F}(q_{i-1})$), but has increased by $\pt{q_0}$ (since $ \p[t]{F}(q_0) = \p[t_0]{F}(q_m)  $).
    Now, recall that $q_0 = q$, and that $(q,q_i)$ was a violation at $t_0$.
    This implies that
    \[ \pt{q_i} < \dstr\ept{q_i} \le \ept{q} \le \pt{q} \]

    proving the second claim.

    Combining the two claims, we have that
    \[ \cov{x}{t}\ge \cov{x}{t^-} \]
    as required.
\end{proof}
\begin{prop}
    \label{prop:UW_NoViolations}
    Outside of function calls, there are no violations in the algorithm.
\end{prop}
\begin{proof}
    Observe that the only changes to $F$ occur in calls to $\UponJobRelease$ and $\UponHeavyF$.
    We show that if there is no violation before such a function call, then there would be no violation after the function call.

    In the case of $\UponHeavyF$, this trivially holds, since removing a job from $F$ cannot cause a new violation.
    It remains to consider $\UponJobRelease$.

    Suppose that $\UponJobRelease{q}$ is called at time $t$.
    As in the proof of  \cref{prop:UW_IrrelevantArrival}, we denote by $t_0$ the time immediately after the insertion of the job $q$ into $F$, and before the possible rotation of \cref{line:UW_Rotate}.
    By the induction hypothesis, and through the fact that $\p[t_0]{F}(q) = \delta_F(t_0)$, the only violations at $t_0$ are of the form $(q,q')$ for some $q' \in Q_F(t_0)\backslash\{q\}$.
    If there are no such violations then no rotation occurs, and thus the state at time $t$ contains no violations, completing the proof.

    Otherwise, denoting $q_0 = q$, the set $Q' = \{ q_0,q_1,\cdots,q_m \}$ is rotated.
    For ease of notation, we write $Q'' = Q_F(t_0) \backslash Q'$.
    Consider the possible violations at time $t$, after this rotation.

    \begin{enumerate}
        \item \textbf{violations of the form $(q''_1,q''_2)$ for some $q''_1,q''_2 \in Q''$.} Such violations cannot occur in $t$, since they did not exist in $t_0$, and since $\p[t]{F}(q'') = \p[t_0]{F}(q'')$ for every $q'' \in Q''$.

        \item \textbf{violations between $q_i$ and $q''$ for some index $0<i \le m$ and some job $q'' \in Q''$.} Consider that there was no violation between $q_i$ and $q''$ at time $t_0$.
        Since $ \p[t]{F}(q'') = \p[t_0]{F}(q'') $, and since $\p[t]{F}(q_i) \ge  \p[t_0]{F}(q_i) $, it must be that the violation is of the form $(q_i,q'')$.

        Now, observe that $\p[t]{F}(q_i) \le \p[t]{F}(q)$.
        In addition, since $(q,q_i)$ was a violation at $t_0$, it holds that $\ept{q_i} < \dstr\ept{q_i} \le \ept{q}$.
        Thus, if $(q_i,q'')$ is a violation at $t$, it must be that $(q,q'')$ was a violation at $t_0$, in contradiction to $q'' \in Q''$.

        \item \textbf{violations of the form $(q_{i},q_{j})$ for $1\le i,j \le m$.} Observe that there was no violation between $q_i$ and $q_j$ at time $t_0$, and that the order between their priorities did not change between $t_0$ and $t$.
        Thus, there is no violation between them at $t$.

        \item \textbf{violations between $q_i$ and $q$, for some index $i$.}  Consider that since $\p[t]{F}(q) \le \p[t]{F}(q_i)$, the violation must be of the form $(q_i,q)$.
        Now, consider that $(q,q_i)$ was a violation at $t_0$, and thus $\dstr\ept{q} > \ept{q} \ge (\ope[q_i]) \ept{q_i} > \ept{q_i}$.
        Thus, $(q_i,q)$ cannot be a violation at $t$.

        \item \textbf{violations between $q$ and $q''$, for some $q'' \in Q''$.} Since $(q,q'')$ was not a violation at $t_0$, and since $\p[t]{F}(q'') = \p[t_0]{F}(q'')$ and $\p[t]{F}(q) \le \p[t_0]{F}(q)$, we have that $(q,q'')$ is not a violation at $t$.
        It remains to consider the possible violation $(q'',q)$.

        Now, consider that $\p[t]{F}(q) = \p[t_0]{F}(q_m)$.
        In addition, since $(q,q_m)$ was a violation at $t_0$, it holds that $\dstr\ept{q_m} \le \ept{q} \le \dstr\ept{q}$.
        Thus, if $(q'',q)$ is a violation at $t$, it must be that $(q'',q_m)$ was a violation at $t_0$, in contradiction to the fact that all violations at $t_0$ involve $q$.
    \end{enumerate}
    Overall, there are no violations at time $t$.
\end{proof}

For ease of notation, we define $\cbc = \ceil{\dstr^2}$.

\begin{figure*}

    \begin{center}
        \subfloat[][\label{subfig:UW_Relevant1}Time $t^-$.]{
        \includegraphics[height=0.3\textheight]{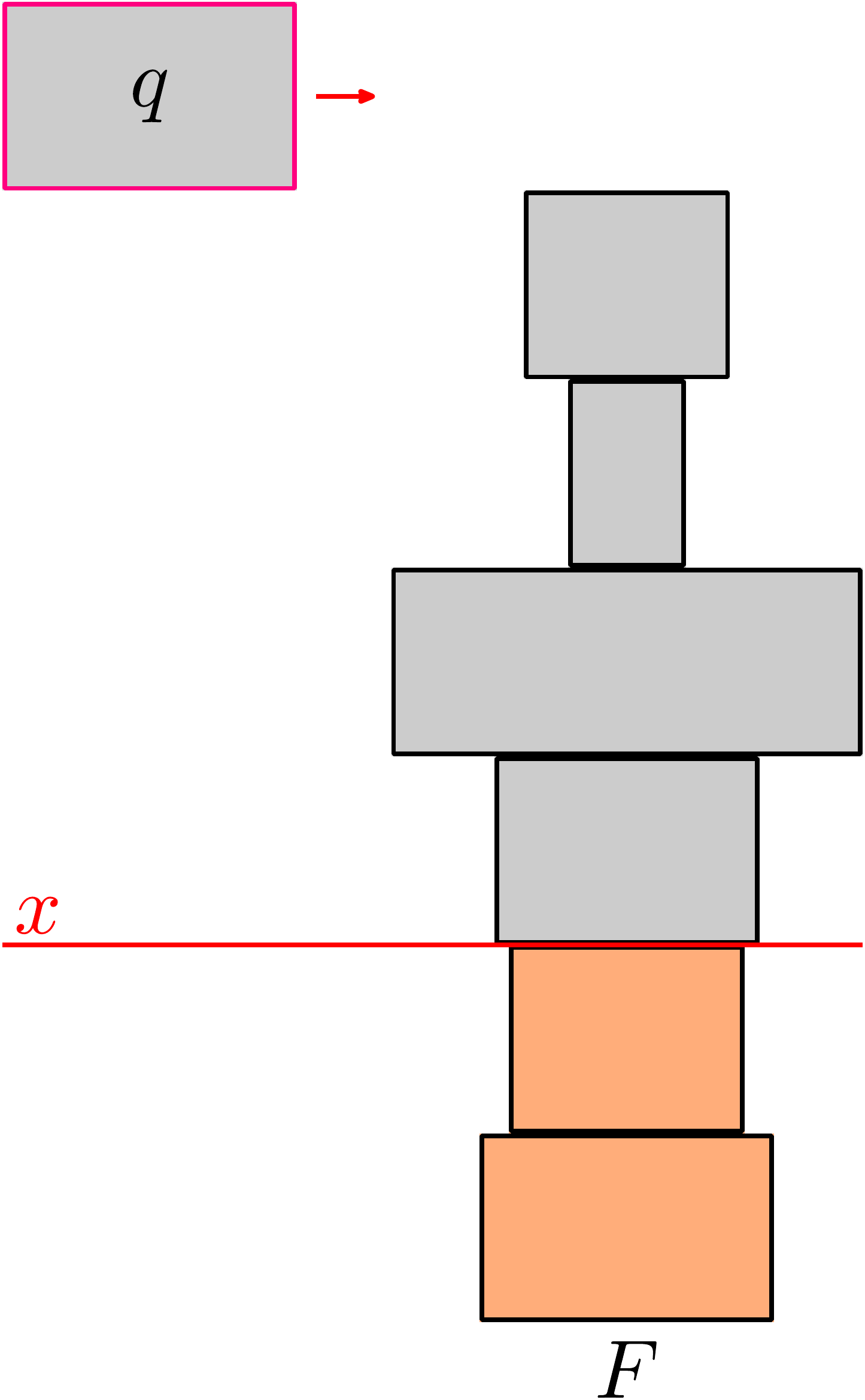}
        }
        \hfill
        \subfloat[][\label{subfig:UW_Relevant2}Time $t_0$.]{
        \hspace{0.1\textwidth}
        \includegraphics[height=0.3\textheight]{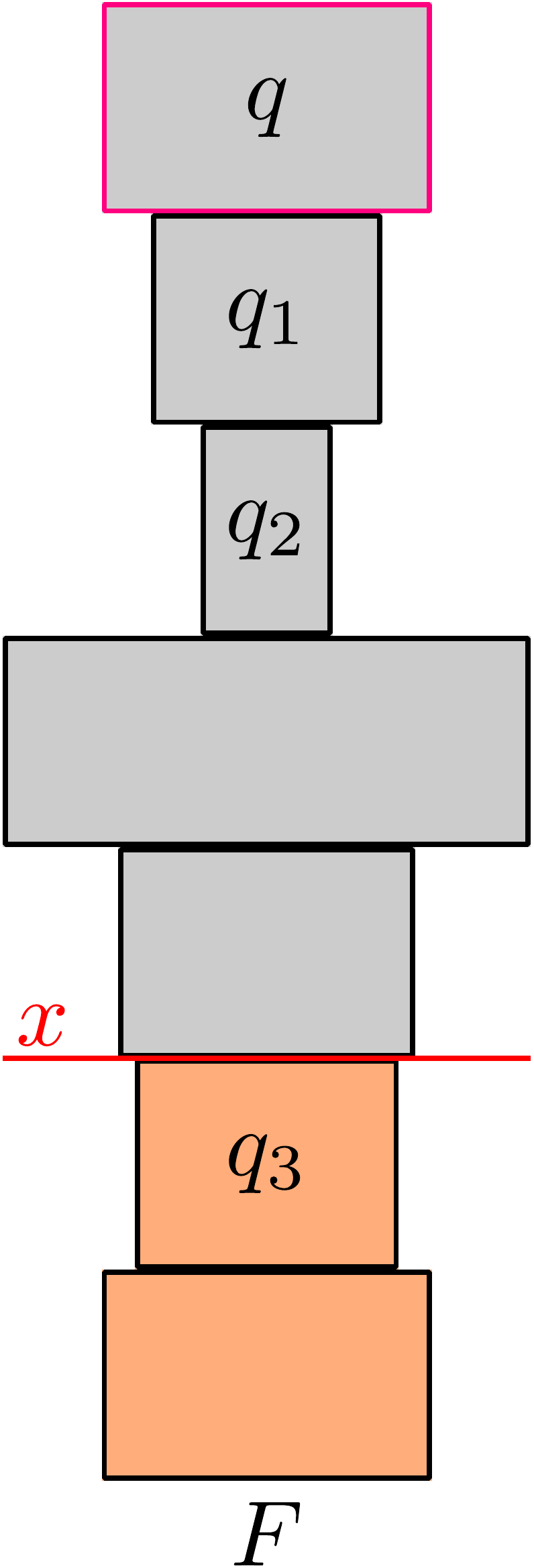}
        \hspace{0.1\textwidth}
        }
        \hfill
        \subfloat[][\label{subfig:UW_Relevant3}Time $t$.]{
        \hspace{0.05\textwidth}
        \includegraphics[height=0.3\textheight]{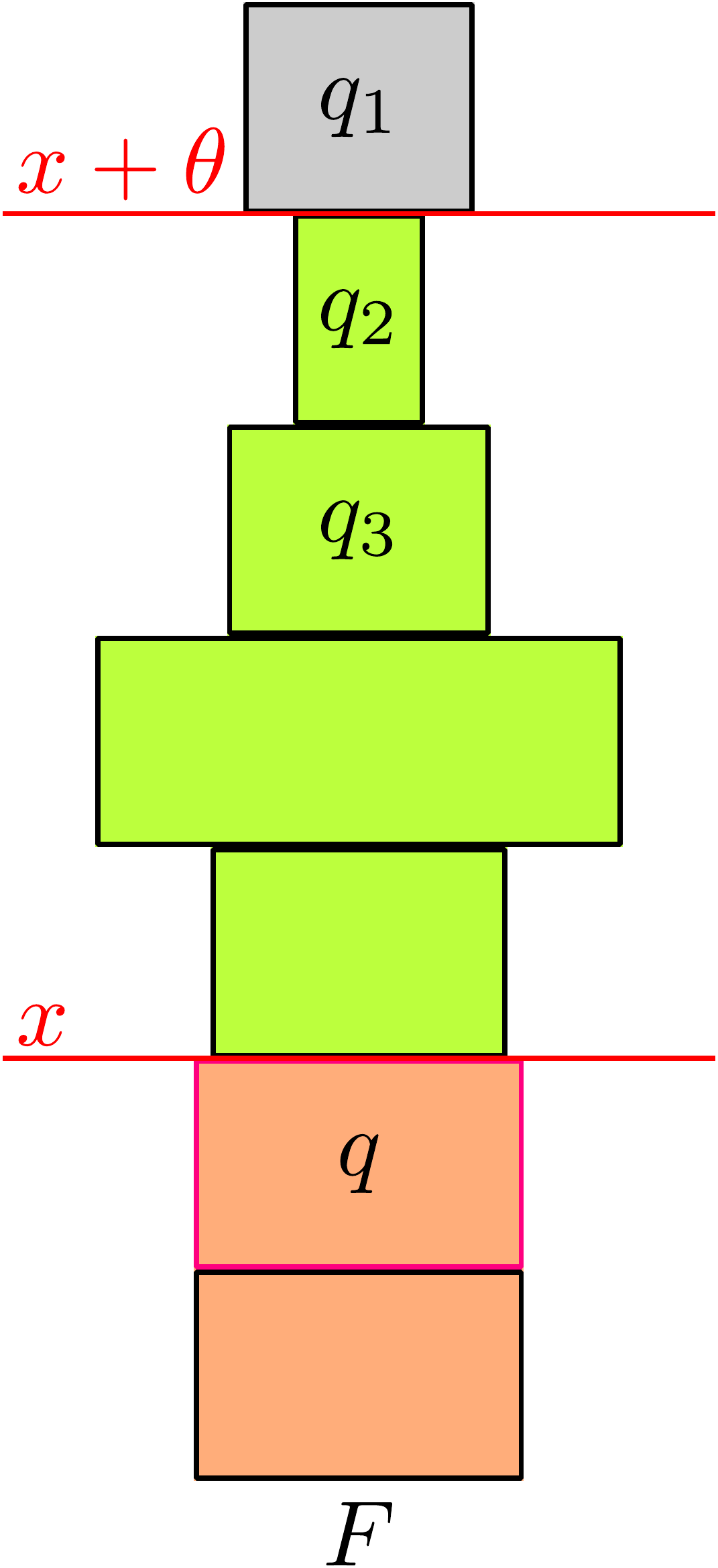}
        \hspace{0.05\textwidth}
        }
    \end{center}
    This figure visualizes \cref{prop:UW_RelevantArrival}, where \cref{subfig:UW_Relevant1,subfig:UW_Relevant2,subfig:UW_Relevant3} show the state of the algorithm at times $t^-,t_0$ and $t$, correspondingly (since no changes occur in bin $P$, it is omitted from the figure).
    Observe that the volume covered by $x$ decreased by $\pt{q_3}$ but increased by $\pt{q}$ from $t_0$ to $t$.
    The claim in the proof of \cref{prop:UW_RelevantArrival} is that the volume gained by raising the bar by $\cbc$ (the green volume in the figure) is at least $\pt{q_3}$.
    \caption[]{Visualization of \cref{prop:UW_RelevantArrival}}
\end{figure*}

\begin{prop}
    \label{prop:UW_RelevantArrival}
    Consider a call to $\UponJobRelease{q}$ at time $t$.
    For every $\xbar{x}$, it holds that
    \[ \cov{x+\cbc}{t} \ge \cov{x}{t^-} + \pt{q} \]
\end{prop}
\begin{proof}
    As in the proofs of \cref{prop:UW_IrrelevantArrival,prop:UW_NoViolations} \cref{prop:UW_IrrelevantArrival,prop:UW_RelevantArrival}, we denote by $t_0$ the time immediately after $q$'s insertion into $F$, and before the possible rotation.
    Consider that  for every $q' \in Q(t) \backslash \{ q \}$ it holds that $\wbase{q'}{t_0} = \wbase{q'}{t^-}$, and thus $\cont{q'}{x}{t_0} = \cont{q'}{x}{t^-}$, which implies that $\cov{x}{t_0} \ge \cov{x}{t^-}$.

    If, in addition to the previous observation, it holds that $x \ge \p[t_0]{F}(q)$, then we have that $\cont{q}{x}{t_0} = \pt{q}$, and thus $ \cov{x}{t_0} \ge \cov{x}{t^-} + \pt{q} $.
    We thus have that
    \[
        \cov{x+\cbc}{t} \ge \cov{x}{t} \ge \cov{x}{t_0} \ge \cov{x}{t^-} + \pt{q}
    \]
    where the second inequality uses Claim 2 of the proof of \cref{prop:UW_IrrelevantArrival}.
    This completes the proof for the case that $x\ge \p[t_0]{F}(q)$, and we thus assume for the remainder of the proof that $x<\p[t_0]{F}(q)$.

    Denote $q_0 = q$, and let $q_1,\cdots,q_m$ be the jobs involved in violations with $q$ at $t_0$ (it is possible that $m=0$ if there are no violations).
    The algorithm performs applies the cyclic permutation $\pr{\p[t_0]{F}(q_m),\cdots,\p[t_0]{F}(q_1),\p[t_0]{F}(q_0)} $ to $\p[t_0]{F}$ to obtain $\p[t]{F}$.

    As in the proof of \cref{prop:UW_IrrelevantArrival}, let $i$ be the minimal index such that $\p[t_0]{F}(q_i) \le x$ (due to the assumption that $\p[t_0]{F}(q) > x$, we have that $i>0$).
    Due to the rotation, the volume covered by $x$ has decreased by $\pt{q_i}$ but increased by $\pt{q_0}$.
    That is, we have that
    \[
        \cov{x}{t} \ge \cov{x}{t_0} - \pt{q_i} + \pt{q_0}
    \]

    Now, it remains to show that
    \begin{equation}
        \label{eq:UW_RelevantArrivalClaim}
        \cov{x+\cbc}{t} \ge \cov{x}{t} + \pt{q_i}
    \end{equation}
    which completes the proof.

    We henceforth write $q^{\star}$ instead of $q_i$.
    It holds that $x < \wbase{q^{\star}}{t}$, and thus $\cont{q^{\star}}{x}{t} = 0$.
    If, in addition, we have that $\wbase{q^{\star}}{t} \le x+\cbc$, then $\cont{q^{\star}}{x+\cbc}{t} = \pt{q^{\star}}$, which implies \cref{eq:UW_RelevantArrivalClaim} and completes the proof.

    Otherwise, we have that $x+\cbc < \wbase{q^{\star}}{t}$.
    This implies that there exist $\cbc$ jobs $q'_1,\cdots,q'_{\cbc}\in Q_F(t)$, such that $ \wbase{q'_j}{t} = \lfloor x+j \rfloor $ for every index $1\le j \le \cbc$.
    Observe that for each index $j$ it holds that $\cont{q'_j}{x}{t} = 0$ and $\cont{q'_j}{x+\cbc}{t} = \pt{q'_j}$.
    Thus, it holds that
    \begin{equation}
        \label{eq:UW_RelevantArrivalRising}
        \cov{x+\cbc}{t} \ge \cov{x}{t} + \sum_{j=1}^{\cbc} \pt{q'_j}
    \end{equation}

    It holds that $\p[t]{F}(q'_j) \le \p[t]{F}(q^{\star})$ for each index $j$.
    Using \cref{prop:UW_NoViolations}, we have
    \[
        \pt{q'_j}
        \ge \ept{q'_j}
        \ge \frac{\ept{q^{\star}}}{\dstr}
        > \frac{\pt{q^{\star}}}{\dstr^2}
        \ge \frac{\pt{q^{\star}}}{\cbc}
    \]

    Plugging into \cref{eq:UW_RelevantArrivalRising} yields \cref{eq:UW_RelevantArrivalClaim}, completing the proof.
\end{proof}
\begin{prop}
    \label{prop:UW_Execution}
    For any $x$ and time $t$, if $\cov{x}{t}$ is decreasing due to the algorithm's processing of a job, then $\cov{x}{t}=V(t)$.
\end{prop}
\begin{proof}
    Suppose that $\cov{x}{t}$ decreased due to the processing of a job $q\in Q_P(t)$ by the algorithm.
    This implies that $\cont{q}{x}{t} > 0$, and thus $x \ge \wbase{q}{t}$.
    But from the choice of job in $\Process$, it holds that $\p[t]{P}(q) = \delta_P(t)$, and thus $\wbase{q}{t} = \delta_P(t)$.
    Therefore, for every $q'\in Q_P(t)$ it holds that $\cont{q'}{x}{t} = \rpt{q'}{t}$.
%

    Moreover, since the algorithm maintains that $\delta_F(t) \le \delta_P(t)$, for every $q'\in Q_F(t)$ we also have that $\cont{q'}{x}{t} = \rpt{q'}{t}$.

    Overall, we have that $\cov{x}{t} = V(t)$, as required.
\end{proof}

\begin{figure}

    \begin{center}
        \includegraphics[height=0.3\textheight]{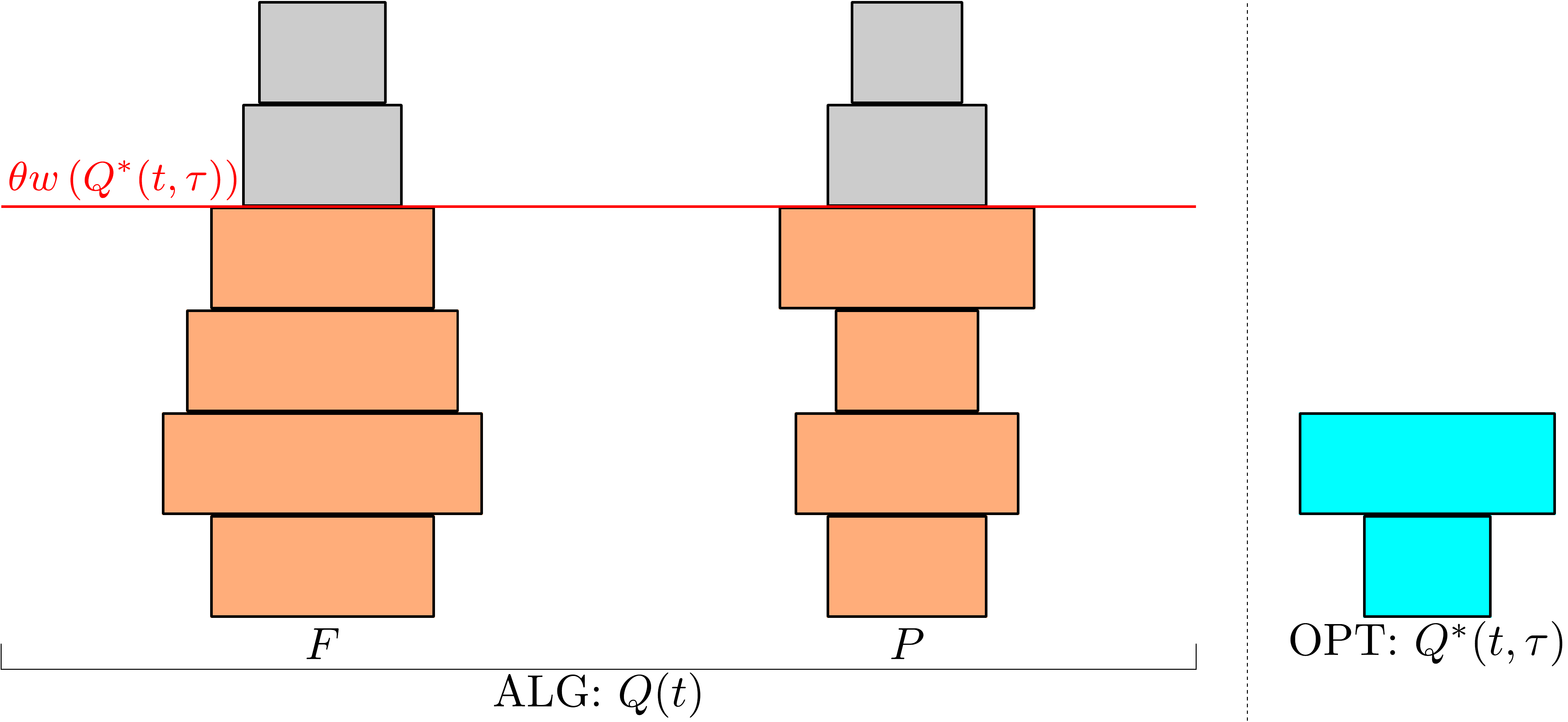}
    \end{center}
    This figure visualizes the inductive claim in the proof of \cref{lem:UW_VolumeLemma}, for some time $\tau \in [0,t]$.
    It shows the state of the algorithm at $\tau$, placed next to the state of $Q^*(t,\tau)$ in the optimal solution at $\tau$.
    The claim is that the total volume covered by $\cbc \wts{Q^*(t,\tau)}$ (the orange volume) is at least the volume of $Q^*(t,\tau)$ in $\opt$ at $\tau$ (the cyan volume).
    \caption[]{\cref{lem:UW_VolumeLemma}}
    \label{fig:UW_VolumeLemma}
\end{figure}

\begin{lem}
    \label{lem:UW_VolumeLemma}
    At any time $t$ it holds that $\cov{\cbc\delta^*(t)}{t} = V(t)$.
\end{lem}

\begin{proof}
    For any time $\tau$, let $Q^*(\tau)$ be the set of jobs pending in the optimal solution at time $\tau$.

    We now fix a time $t$.
    Let $\tau$ be any time in the range $[0,t]$.
    We also define $Q^*(t,\tau) = Q^*(t)\cap Q^*(\tau)$ -- i.e. the jobs alive at $t$ which were already released by $\tau$.
    For ease of notation we write $Y^*(\tau) = \sum_{q\in Q^*(t,\tau)} \rpto{q}{\tau}$.

    We prove, by induction on $\tau$, that for every $\tau\in [0,t]$ it holds that
    \begin{equation}
        \label{eq:UW_InductiveClaim}
        \cov{\cbc\cdot \ps{Q^*(t,\tau)}}{\tau} \ge Y^*(\tau)
    \end{equation}

    A visualization of the claim in \cref{eq:UW_InductiveClaim} appears in \cref{fig:UW_VolumeLemma}.

    Clearly, \cref{eq:UW_InductiveClaim} holds for $\tau=0$, as $Y^*(\tau)=0$.
    Now, we show that no possible event can break the inequality of \cref{eq:UW_InductiveClaim} as time $\tau$ progresses.
    When considering an event at $\tau$, we denote by $\tau^-$ the time immediately before the event.
    Consider the possible events:
    \begin{enumerate}
        \item \textbf{A job moves from $F$ to $P$ in the algorithm.} \cref{prop:UW_JobMoveRetainsVolume} implies that the left-hand side of \cref{eq:UW_InductiveClaim} does not decrease upon this event.
        Since the right-hand side does not change, the inequality continues to hold.

        \item \textbf{A job $q\notin Q^*(t)$ is released.} In this case, the right-hand side of \cref{eq:UW_InductiveClaim} remains the same. \cref{prop:UW_IrrelevantArrival} implies that the left-hand side does not decrease.

        \item \textbf{A job $q\in Q^*(t)$ is released.} In this case, the right-hand side of \cref{eq:UW_InductiveClaim} increases by $\pt{q}$, as $Q^*(t,\tau^-) = Q^*(t,\tau) \cup \{ q \}$. \cref{prop:UW_RelevantArrival} implies that
        \[
            \cov{\cbc\cdot \ps{Q^*(t,\tau)}}{\tau} =
            \cov{\cbc\cdot \ps{Q^*(t,\tau^-)}+ \cbc}{\tau} \ge
            \cov{\cbc\cdot \ps{Q^*(t,\tau^-)}}{\tau^-} + \pt{q}
        \]
        and thus the inequality of \cref{eq:UW_InductiveClaim} continues to hold.

        \item \textbf{A job $q\in Q(\tau)$ is processed.} \cref{prop:UW_Execution} implies that if the left-hand side decreases as a result of processing, then $\cov{\cbc\cdot \ps{Q^*(t,\tau)}}{\tau} = V(\tau)$.
        Since the algorithm is not lazy (i.e. always processes a pending job if there exists one), it holds that $V(\tau) = V^*(\tau) \ge Y^*(\tau) $, and thus the inequality holds.

    \end{enumerate}

    The proof of the induction claim is complete.
    Now observe that choosing $\tau=t$, we have that $\ps{Q^*(t,\tau)} = \delta^*(t)$ and that $Y^*(t) = V^*(t) = V(t)$.
    Thus, \cref{eq:UW_InductiveClaim} yields that
    \[ \cov{\cbc \delta^*(t)}{t} = V(t) \]
    completing the proof.
\end{proof}

\begin{proof}[Proof of \cref{thm:UW_Competitiveness}]
    \Cref{lem:UW_VolumeLemma} implies that at any time $t$ we have $\cov{\cbc \delta^*(t)}{t} = V(t) $.
    This implies that $\delta_F(t) \le \cbc\delta^*(t)$; otherwise, there would exist a pending job $q\in Q_F(t)$ such that $\wbase{q}{t} > \cbc\delta^*(t)$, the volume of which would not be counted in $\cov{\cbc \delta^*(t)}{t}$.
    The same argument applies to $P$ as well, yielding that $\delta_P(t) \le \cbc\delta^*(t)$.

    Overall, we get that $\delta(t)\le 2\cbc\delta^*(t)$.
    Integrating over $t$ completes the proof of the theorem.

\end{proof}

\subsection{Proof of \cref{thm:LogW_Competitiveness}}
\label{subsec:LogW_Proofs_W}

In the following analysis, we prove \cref{thm:LogW_Competitiveness}.

\begin{defn}[time-dependent variable values]
    For any time $t$, we define $Q_{F^i}(t)$ and $Q_{P^i}(t)$ to be the values of the variables $Q_{F^i}$ and $Q_{P^i}$ at time $t$, respectively.
    We define $Q_{A^i}(t) = Q_{F^i}(t)\cup Q_{P^i}(t)$, which is the set of pending jobs in superbin $A^i$ at time $t$.

    We also define $\p[t]{F^i}$ and $\p[t]{P^i}$ to be the values of the priority-mapping variables $\p{F^i}$ and $\p{P^i}$ at time $t$.
\end{defn}

\begin{defn}[base of a job]
    For every time $t$ and every job $q\in Q_{F^i}(t)$, we define the \emph{base} of $q$ at time $t$ as follows:
    \[
        \base{q}{t} = \pc {q' \in Q_{F^i}(t) \middle| \p[t]{F^i}(q')\le \p[t]{F^i}(q) }
    \]
    For ease of notation, we denote $w(\base{q}{t})$ by $\wbase{q}{t}$.
    For every job $q\in Q_{P^i}(t)$, we similarly define $\base{q}{t}$ and $\wbase{q}{t}$ according to $\p[t]{P^i}$.
\end{defn}

%
%
%

We now redefine the volume covered by a bar $x$.
\begin{defn}[Volume covered by bar]
    \label{defn:LogW_CoveredVolume}
    For any $\xbar{x}$ and time $t$, and for any job $q\in Q(t)$, we define \[ \cont{q}{x}{t} =
    \begin{cases}
        \rpt{q}{t} & x\ge \wbase{q}{t}\\
        0 & \text{otherwise}

    \end{cases}
    \]
    called the volume of $q$ covered by $x$ at $t$.

    For every superbin index $i$, we define the total volume covered by $x$ in various bins:
    \begin{itemize}\itemsep0em
    \item $\cov[F^i]{x}{t} = \sum_{q\in Q_{F^i}(t)} \cont{q}{x}{t}$

    \item $\cov[P^i]{x}{t} = \sum_{q\in Q_{P^i}(t)} \cont{q}{x}{t}$

    \item $\cov[A^i]{x}{t} = \cov[F^i]{x}{t} + \cov[P^i]{x}{t}$

    \item $\cov{x}{t} = \sum_{i} \cov[A^i]{x}{t}$
    \end{itemize}
\end{defn}

\cref{fig:LogW_CoveredVolume} visualizes \cref{defn:LogW_CoveredVolume}, where $\cov{x}{t}$ is the total orange area in the figure.

\begin{figure}
    \begin{center}
        \includegraphics[height=0.3\textheight]{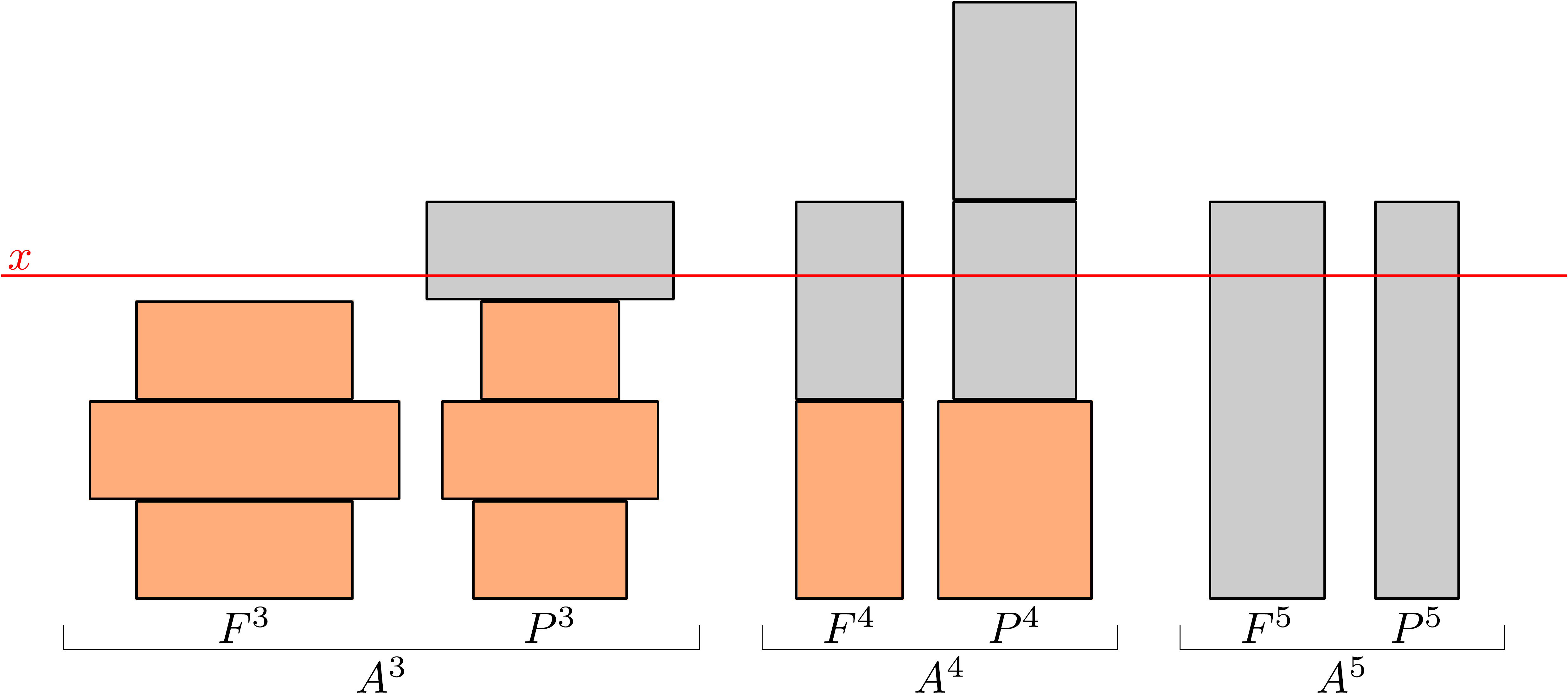}
    \end{center}

    \caption[]{Visualization of Covered Volume}
    \label{fig:LogW_CoveredVolume}
\end{figure}

For ease of notation, we again define $\cbc = \ceil{(\ope)^2}$.
\begin{prop}
    \label{prop:LogW_Arrivals}
    Upon the release of a job $q$ at time $t$, and for every $\xbar{x}$, it holds that
    \begin{itemize}\itemsep0em
    \item 	$ \cov{x}{t} \ge \cov{x}{t^-}  $

    \item 	$ \cov{x+\cbc\wt{q}}{t} \ge \cov{x}{t^-} + \pt{q}  $
    \end{itemize}
\end{prop}
\begin{proof}
    Let $i$ be such that $\wt{q} = 2^i$.
    The job $q$ is thus sent to the superbin $A^i$.

    To prove the first claim, first note that the structure of every superbin $A^{i'}$ for $i'\neq i$ remains the same after $q$'s release.
    Thus, we have that $\cov[A^{i'}]{x}{t} = \cov[A^{i'}]{x}{t^-}$.
    As for the volume covered in $A^i$, we know that it does not decrease due to \cref{prop:UW_IrrelevantArrival} (it is easy to verify that this proposition holds when all weights are scaled up by $2^i$).

    To prove the second item, again note that for every $i'\neq i$ it holds that $\cov[A^{i'}]{x+\cbc\wt{q}}{t} \ge \cov[A^{i'}]{x}{t} = \cov[A^{i'}]{x}{t^-}$.
    As for the volume covered in $A^i$, using \cref{prop:UW_RelevantArrival} (scaled up by $2^i$) yields that $\cov[A^i]{x+\cbc\wt{q}}{t} \ge \cov[A^i]{x}{t^-} + \pt{q}$.
    This completes the proof.
\end{proof}
\begin{prop}
    \label{prop:LogW_Execution}
    For any $x$ and time $t$, if $\cov{x}{t}$ is decreasing due to the algorithm's processing of a job, then $\cov{x}{t}=V(t)$.
\end{prop}
\begin{proof}
    Consider the job $q$ being processed, where $q\in P^i$ for some $i$.
    Since $\cont{q}{x}{t} > 0$, we have that $x\ge \wbase{q}{t}$.
    Since the algorithm processes the job $q\in Q_{P^i}(t)$ such that $\p[t]{P^i}(q) = \delta_{P^i}(t)$, we have that $\wbase{q}{t} = \delta_{P^i}(t) \cdot 2^i = w(Q_{P^i}(t))$.

    Now, observe that since the algorithm chose superbin $A^i$ for processing, it must be that for every $i'$ we have $w(Q^i_P(t)) \ge w(Q^{i'}_P(t)) \ge w(Q^{i'}_F(t))$, where the second inequality is since the algorithm moves jobs from the full bins to the partial bins.

    Since $x$ is at least the weight of every bin in the system, it is at least $\wbase{q'}{t}$ for every $q'\in Q(t)$, and thus $\cov{x}{t} = V(t)$.
\end{proof}

\begin{lem}
    \label{lem:LogW_VolumeLemma}
    At any time $t$ it holds that $\cov{\cbc\cdot w(Q^*(t))}{t} = V(t)$.
\end{lem}
\begin{proof}
    The proof of this lemma is nearly identical to that of \cref{lem:UW_VolumeLemma} of \cref{subsec:UW}, but uses the propositions of this subsection instead.

    We now fix a time $t$.
    Let $\tau$ be any time in the range $[0,t]$.
    We again define $Q^*(t,\tau)$ and $Y^*(\tau)$ as in the proof of \cref{lem:UW_VolumeLemma}.

    We prove, by induction on $\tau$, an identical claim to that in the proof of \cref{lem:UW_VolumeLemma}, which is that for every $\tau\in [0,t]$ it holds that
    \begin{equation}
        \label{eq:LogW_InductiveClaim}
        \cov{\cbc\cdot \ps{Q^*(t,\tau)}}{\tau} \ge Y^*(\tau)
    \end{equation}

    The claim trivially holds for time $\tau=0$, before any requests were released.
    We now consider the possible events which could possibly break the claim as $\tau$ progresses, in a similar way to \cref{lem:UW_VolumeLemma}.

    \begin{enumerate}
        \item \textbf{A job moves from $F$ to $P$ in the algorithm.} \Cref{prop:UW_JobMoveRetainsVolume} implies that the left-hand side of \cref{eq:LogW_InductiveClaim} does not decrease upon this event.
        Since the right-hand side does not change, the inequality continues to hold.

        \item \textbf{A job $q\notin Q^*(t)$ is released.} In this case, the right-hand side of \cref{eq:LogW_InductiveClaim} remains the same.
        The first claim in \cref{prop:LogW_Arrivals} implies that the left-hand side does not decrease.

        \item \textbf{A job $q\in Q^*(t)$ is released.} In this case, the right-hand side of \cref{eq:LogW_InductiveClaim} increases by $\pt{q}$, as $Q^*(t,\tau^-) = Q^*(t,\tau) \cup \{ q \}$.
        The second claim in \cref{prop:LogW_Arrivals} implies that
        \[
            \cov{\cbc\cdot \wts{Q^*(t,\tau)}}{\tau} =
            \cov{\cbc\cdot \wts{Q^*(t,\tau^-)} + \cbc\wt{q}}{\tau} \ge
            \cov{\cbc\cdot \wts{Q^*(t,\tau^-)}}{\tau^-} + \pt{q}
        \]
        and thus the inequality of \cref{eq:LogW_InductiveClaim} continues to hold.

        \item \textbf{A job $q\in Q(\tau)$ is processed.} \Cref{prop:LogW_Execution} implies that if the left-hand side decreases as a result of processing, then $\cov{\cbc\cdot \ps{Q^*(t,\tau)}}{\tau} = V(\tau)$.
        Since the algorithm is not lazy (i.e. always processes a pending job if there exists one), it holds that $V(\tau) = V^* (\tau) \ge Y^*(\tau) $, and thus the inequality holds.

    \end{enumerate}

    The proof of the induction claim is now complete, and choosing $\tau=t$ in \cref{eq:LogW_InductiveClaim} yields
    \[ \cov{\cbc \delta^*(t)}{t} = V(t) \]
    as required.
\end{proof}

\begin{proof}[Proof of \cref{thm:LogW_Competitiveness}]
    At any time $t$, the algorithms contains at most  $\ceil{\log W}+1$ nonempty superbins, each of which contains two bins (full and partial).
    From \cref{lem:LogW_VolumeLemma}, at any time $t$, there does not exist a pending job $q$ in the algorithm such that $\wbase{q}{t} > \cbc\cdot \wts{Q^*(t)}$.
    Thus, the total weight in each bin is at most $\cbc\cdot \wts{Q^*(t)}$, and thus $2\cbc\cdot \wts{Q^*(t)}$ for each superbin (which contains two bins).
    Overall, the total weight in the system is at most $2\cbc (\ceil {\log W} +1) \wts{Q^*(t)}$.

    Integrating over $t$ thus yields the desired result, which is that the algorithm is $O(\cbc\log W)$ competitive.
\end{proof}

\section{Lower Bound for Unweighted \prob}
\label{sec:UWLB}
In this section, we show a lower bound of $2$ on the competitive ratio of any algorithm for \prob, even when the distortion parameter is arbitrarily close to $1$.
This lower bound is perhaps somewhat surprising, as one would hope that as $\dstr$ approaches $1$, one would be able to approximate the optimal SRPT schedule, and thus approach $1$-competitiveness.

\begin{thm}
    \label{thm:UWLB_LowerBound}
    For any $\dstr > 1$, there is no deterministic $\dstr$-robust algorithm with competitive ratio less than $2$.
\end{thm}
\begin{proof}
    Fix some $\dstr > 1$.
    We assume that $\dstr \le 2$, as we are interested in the case of $\dstr$ approaching $1$ (clearly if the lower bound holds for small $\dstr$, it holds for large $\dstr$).

    Let $\lambda = \frac{\dstr+1}{\dstr - 1}$, and let $M$ be an arbitrarily large integer.
    The adversary we describe operates in $M$ phases, numbered from $M-1$ down to $0$.
    Phase $i$ takes $\lambda^i$ time units.
    Observe the following adversary:
    \begin{enumerate}
        \item for $i$ from $M-1$ to $0$:
        \begin{enumerate}
            \item release two jobs of predicted processing time $\lambda^i$.

            \item wait $\lambda^i$ time units.

            \item of the two jobs released in this phase, let $q_1^i$ be the job processed more by the algorithm during these $\lambda^i$ time units, and let $q_2^i$ be the other job.
            The adversary sets $\pt{q_1^i} = \dstr \lambda^i$ and $\pt{q_2^i} = \lambda^i$.
        \end{enumerate}
    \end{enumerate}

    During phase $i$, the algorithm has processed job $q_1^i$ for at most $\frac{\lambda^i}{2}$ time, and has processed job $q_2^i$ for at most $\lambda^i$ time.
    Thus, both $q_1^i$ and $q_2^i$ have remaining volume of $(\dstr-1)\lambda^i = (\dstr+1)\lambda^{i-1}$ at the end of the phase.
    Also observe that the total time of phases $i-1$ through $0$ is at most
    \[
        \sum_{i'=0}^{\infty}\lambda^{i-1-i'} = \lambda^{i-1} \cdot \frac{1}{1-\lambda^{-1}}=\lambda^{i-1} \cdot \frac{1-\dstr}{2}
    \]
    And thus, at the end of the last phase, it must be that the remaining processing time for both $q_1^i$ and $q_2^i$ in the algorithm is at least $\frac{1-\dstr}{2}\lambda^{i-1}$.

    Overall, we have that the algorithm has $2M$ living jobs at the end of the $M$ phases, each with at least $\frac{1-\dstr}{2} \lambda^0 = \frac{1-\dstr}{2}$ remaining volume.
    The optimal solution, on the other hand, has $M$ living jobs, as it would complete $q_2^i$ in each phase $i$.

    Now, after the $M$ phases, the adversary would begin the ``bombardment'', releasing a job of processing time $x = \frac{1-\dstr}{2}$ every $x$ time units.
    The algorithm would have no option which is better than serving the ``bombardment'' requests, and the optimal solution would do the same.
    During the bombardment, the algorithm has $2M+1$ living jobs at any time, and the optimal solution has $M+1$.
    Thus, the competitive ratio tends to $\frac{2M+1}{M+1}$ as the bombardment continues.
    Since $M$ was chosen arbitrarily, we can let $M$ tend to $\infty$, and thus the competitive ratio tends to $2$.
\end{proof}

\section{Semiclairvoyant Scheduling}
\label{sec:SCS}

In the semiclairvoyant model, we are given the logarithmic class of a job rather than an estimate for its processing time.
That is, for a job $q$ the algorithm is given $\cls{q} = \floor {\log_{\rho} \pt{q}} $ and not $\pt{q}$ (for a constant $\rho > 1$).

Clearly, this model can be reduced to the prediction model as follows.
Upon the arrival of $q$, define $\ept{q} = \rho^{\cls{q}}$ and $\dstr = \rho$.
From the definition of class, it is clear that indeed $\pt{q} \in \IR{\ept{q}}{\dstr\ept{q}}$.

Denote the semiclairvoyant model with the parameter $\rho$ by $\SC[\rho]$.
Applying the theorems of this paper for \prob, we immediately obtain the first results for weighted flow time in the semiclairvoyant setting.

\begin{cor}[of \cref{thm:LogW_Competitiveness,thm:CLP_Competitiveness,thm:CLP_LogDCompetitiveness}]
    In the $\SC[\rho]$ model, there exist:
    \begin{enumerate}
        \item An $O(\log P)$-competitive algorithm (using \cref{thm:CLP_Competitiveness})
        \item An $O(\log D)$-competitive algorithm (using \cref{thm:CLP_LogDCompetitiveness})
        \item An $O(\log W)$-competitive algorithm (using \cref{thm:LogW_Competitiveness})
    \end{enumerate}
\end{cor}

Note that these results for the semiclairvoyant setting match the best known results for the clairvoyant setting, in terms of all three parameters $P,D,W$.

In addition, we consider the algorithm in \cref{subsec:UW} for unweighted \prob, and show that for the $\SC[\rho]$ model it obtains an improved competitive ratio, which is $2\ceil{\rho}$.
In particular, when $\rho = 2$ this algorithm is $4$ competitive, improving upon the $13$-competitive algorithm of~\cite{DBLP:journals/tcs/BecchettiLMP04}.

\subsection{Unweighted Semiclairvoyant Scheduling}
\label{subsec:SCS_UW}

We now show the following theorem for the competitiveness of the algorithm of \cref{subsec:UW} for the unweighted semiclairvoyant setting.
\begin{thm}
    \label{thm:SCS_UnweightedCompetitiveness}
    In the $\SC[\rho]$ model, the competitive ratio of the algorithm of \cref{subsec:UW} is $2\ceil{\rho}$.
\end{thm}

\begin{cor}
    \label{cor:SCS_UnweightedBase2Competitiveness}
    In the $\SC[2]$ model (i.e. base-2 classes) the algorithm of \cref{subsec:UW} is $4$-competitive.
\end{cor}

As observed before, we have that $\dstr = \rho$.

For ease of notation, we define $\cbp = \ceil{\dstr} = \ceil{\rho}$.
The proof of \cref{thm:SCS_UnweightedCompetitiveness} is identical to that of \cref{thm:UW_Competitiveness}, except for replacing \cref{prop:UW_RelevantArrival} with the following proposition.

\begin{prop}[stronger version of \cref{prop:UW_RelevantArrival}]
    \label{prop:SCS_RelevantArrival}
    Consider a call to $\UponJobRelease{q}$ at time $t$.
    For every $\xbar{x}$, it holds that
    \[ \cov{x+\cbp}{t} \ge \cov{x}{t^-} + \pt{q} \]
\end{prop}
\begin{proof}
    The proof is identical to the proof of \cref{prop:UW_RelevantArrival}, except for replacing the claim that
    \[ \cov{x+\cbc}{t} \ge \cov{x}{t} + \pt{q'}  \]
    with the stronger claim that
    \[ \cov{x+\cbp}{t} \ge \cov{x}{t} + \pt{q'}  \]
    with $q'$ defined as in the proof of \cref{prop:UW_RelevantArrival}.

    From the definition of $q'$ in the original proof, it holds that $\wbase{q'}{t} > x$, and thus $\cont{q'}{x}{t} = 0$.
    If $\wbase{q'}{t} \le x+\cbp$, then $\cont{q'}{x+\cbp}{t} = \pt{q'}$, and the proof is complete.

    Otherwise, it holds that $\wbase{q'}{t} > x+\cbp$.
    Consider the $\cbp$ distinct jobs $q_1,\cdots,q_{\cbp} \in Q_F(t)$ such that $\wbase{q_j}{t} = x+j$ for each index $j$.
    Observe that for each index $j$ it holds that $\cont{q_j}{x}{t} = 0$ and $\cont{q_j}{x+\cbp}{t} = \pt{q_j}$.
    Thus, it holds that
    \begin{equation}
        \label{eq:SCS_RelevantArrivalRising}
        \cov{x+\cbp}{t} \ge \cov{x}{t} + \sum_{j=1}^{\cbp} \pt{q_j}
    \end{equation}

    It holds that $\p[t]{F}(q') > \p[t]{F}(q_j)$ for each $j$.
    Using \cref{prop:UW_NoViolations}, we have that the state of the algorithm at time $t$ contains no violations.
    However, observe that if $\cls{q_j} <\cls{q'}$, then we have that

    \[
        \rho\ept{q_j} = \rho\cdot \rho^{\cls{q_j}}=\rho^{\cls{q_j}+1} \le \rho^{\cls{q'}} = \ept{q'}
    \]
    which is a violation, in contradiction to \cref{prop:UW_NoViolations}.
    Thus, for every $j$ it holds that $\cls{q_j} \ge \cls{q}$, and thus $\pt{q_j} \ge \ept{q_j} \ge \rho^{\cls{q_j}} \ge \rho^{\cls{q'}} \ge \frac{\pt{q'}}{\rho} \ge \frac{\pt{q'}}{\cbp}$.

    Plugging into \cref{eq:SCS_RelevantArrivalRising}, we get
    \[ \cov{x+\cbp}{t} \ge \cov{x}{t} + \pt{q'} \]
    completing the proof.
\end{proof}

This completes the proof of \cref{thm:SCS_UnweightedCompetitiveness}.

\section{Proof of \cref{thm:CLP_LogDCompetitiveness}}
\label{sec:LogD_Proofs}
This section proves \cref{thm:CLP_LogDCompetitiveness}.

Observe that the maximum ratio of \emph{estimated} densities (i.e.~$\frac{\wt{q}}{\ept{q}}$) with respect to the original weight (before rounding) is at most $\dstr D$.
After rounding up to powers of $\bround$, which is $\Theta(\dstr)$, we have that the maximum ratio of densities becomes $O(\dstr^2 D)$.
Hence, the maximum ratio of estimated \id values is also $O(\dstr^2 D)$.
We thus define the number of \id classes to be $\nclass'$, and observe that
  $\nclass' = O(\log (\dstr^2 D) = O(\log (\dstr D))$.

The following lemma immediately implies \cref{thm:CLP_LogDCompetitiveness}, in the same manner in which \cref{lem:CLP_LocalCompetitiveness} implied \cref{thm:CLP_Competitiveness}.

\begin{lem}[analogue of \cref{lem:CLP_LocalCompetitiveness}]
    \label{lem:LogD_LocalCompetitiveness}
    At any point in time $t$, it holds that $\Wt{t} \le O(\nclass' \dstr^2) \cdot \Wt*{t}$.
\end{lem}

\begin{proof}[Proof of \cref{lem:LogD_LocalCompetitiveness}]
    As noted in the proof of \cref{lem:CLP_LocalCompetitiveness}, we have $\clw{t}\le \owc$.
    Define $j_{\min}$ and $j_{\max}$ to be the minimum and maximum \id classes, respectively.
    We have
    \begin{align*}
        \Wt{t}
        &= \sum_{j=j_{\min}}^{j_{\max}} \Wtp{\le \owc}{=j}{t} \\
        &\le 2(j_{\max}-j_{\min}+1) \qw{\owc} + 2\dstr\Wt*{t} + \sum_{j=j_{\min}}^{j_{\max}} 2\max\pc{0, \frac{\Delta \Vp{\le \owc}{\le j}{t}}{2^j}}\\
        &\le 2\nclass' \qw{\owc} + 2\dstr\Wt*{t} + \sum_{j=j_{\min}}^{j_{\max}} 8\dstr\cdot \qw{\clw{\lastt{j}}}
    \end{align*}
    using the same arguments as in the proof of \cref{lem:CLP_LocalCompetitiveness}.
    Now, observe that for every \id class $j$ we have $\clw{\lastt{j}} \le \owc$ from \cref{prop:CLP_DeltaAndWeightBounds}.
    Thus
    \begin{align*}
        \Wt{t} &\le (8\dstr+2)\nclass' \qw{\owc}+ 2\dstr \Wt*{t} \le O(\nclass'\dstr^2)\Wt*{t}\qedhere
    \end{align*}
\end{proof}


\end{document}